\documentclass[a4paper,onecolumn,11pt,accepted=2025-12-30]{quantumarticle}
\pdfoutput=1
\usepackage[margin=1in]{geometry}
\usepackage[utf8]{inputenc}
\usepackage{graphicx}
\usepackage{amsmath}
\usepackage{amssymb}
\usepackage[colorlinks=true, linkcolor=blue, citecolor=green]{hyperref}
\usepackage{amsfonts}
\usepackage{amsthm}
\usepackage{physics}
\usepackage{bbold}
\usepackage{indentfirst}
\usepackage[dvipsnames,usenames]{xcolor}
\usepackage{mathrsfs}
\usepackage{tikz} 
\usetikzlibrary{quantikz}
\usepackage{apptools}
\usepackage{subfig}
\usepackage{multirow}
\usepackage[dvipsnames]{xcolor}
\AtAppendix{\counterwithin{theorem}{section}}
\usepackage{bm}
\usepackage{tikz}
\usetikzlibrary{quantikz}

\newtheorem{theorem}{Theorem}
\newtheorem{lemma}[theorem]{Lemma}
\newtheorem{corollary}[theorem]{Corollary}

\newcommand{\UW}{{Department of Physics, University of Washington, Seattle, WA 98195, USA}}

\newcommand{\UNITN}{{Dipartimento di Fisica, University of Trento, via Sommarive 14, I–38123, Povo, Trento, Italy}}
\newcommand{\TIFPA}{INFN-TIFPA Trento Institute of Fundamental Physics and Applications,  Trento, Italy}
\newcommand{\UT}{{Department of Computer Science, University of Toronto, Toronto, ON M5S 2E4, Canada}}
\newcommand{\PNNL}{{Pacific Northwest National Laboratory, Richland, WA 99354, USA}}

\begin{document}

\title{Fault-tolerant simulation of Lattice Gauge Theories with gauge covariant codes}

\author{Luca~Spagnoli}
\email{luca.spagnoli@unitn.it}
\affiliation{\UNITN}
\affiliation{\TIFPA}
\author{Alessandro~Roggero}
\email{a.roggero@unitn.it}
\affiliation{\UNITN}
\affiliation{\TIFPA}

\author{Nathan~Wiebe}
\email{nawiebe@cs.toronto.edu}
\affiliation{\UT}
\affiliation{\PNNL}
\affiliation{\UW}

\date{2025-12-30}

\begin{abstract}
    We show in this paper that a strong and easy connection exists between quantum error correction and Lattice Gauge Theories (LGT) by using the Gauge symmetry to construct an efficient error-correcting code for an Abelian $\mathbb{Z}_2$ LGT.  We identify the logical operations on this gauge covariant code and show that the corresponding Hamiltonian can be
    expressed in terms of these logical operations
    while preserving the spatial locality of the interactions.
    Furthermore, we demonstrate that these substitutions actually give a new way of writing the LGT as an equivalent hardcore boson model.  Finally we demonstrate a method to perform fault-tolerant time evolution of the Hamiltonian within the gauge covariant code using both product formulas and qubitization approaches.  This opens up the possibility of inexpensive end to end dynamical simulations that save physical qubits by blurring the lines between simulation algorithms and quantum error correcting codes.
\end{abstract}
\maketitle

\section{Introduction}
The field of quantum computation keeps making big steps toward platforms that will be able to perform some calculations with an exponential speedup compared to classical computers. An important example of such calculations could be the simulation of Standard Model physics using Lattice Gauge Theories (LGTs), where classical computers struggle with an exponential complexity, while quantum computers are a promising way to solve the problem. Even though a formal proof that the entire gauge theory describing the Standard Model could be simulated on quantum computers with an exponential speedup is still missing, hints that this might be the case are available since the pioneering work on scalar field theory where the computation of scattering amplitudes can be shown to be BQP-complete~\cite{JLP2012,Jordan2018bqpcompletenessof}. In the past years, many efforts have been made in order to design efficient algorithms for the digital quantum simulations of LGTs for a variety of phenomenologically important gauge groups (for recent reviews see~\cite{Ba_uls_2020,Zohar_2021,Klco_2022,Bauer_PRX_2023})

An important challenge to overcome in order to enable practical large scale simulations using quantum computing platforms is the efficient and scalable control of hardware noise that, if left undisturbed, can lead to a catastrophic accumulation of errors in a quantum simulation. A possible solution to this challenge, inspired by classical error-correcting codes, has emerged since the early days of quantum computing and allows for controllable, and efficient, long time computations by means of Quantum Error Correction (QEC) \cite{ShorQEC,Gottesman1997, Gottesman1998,Knill_qec_98,Terhal2015,Roffe2019}. Much like classical error correction, in QEC one exploits a redundant description of the calculation in order to detect and correct localized errors that might appear during the execution. Since physical theories that describe both fermions as well as gauge fields on a lattice have natural redundancies generated by the presence of symmetries, they provide a fertile ground for exploring special purpose error correction strategies tailored to quantum simulations of these models. Indeed several recent works have explored the possibility of designing error detection and error correction protocols for these physical theories with the goal of both improving the performance of error correction but also as a mean to understand the properties of gauge theories from the point of view of QEC~\cite{Stryker2019,Raychowdhury2020,Faist2020,Klco2021_qec,Kong2022,Rajput2023,delPino2023,gustafson2023robustness,Chen2024,carena2024quantum, wauters2024symmetryprotection}. Inspired by the work of Stryker on error detection oracles for gauge theories employing a direct check of local Gauss' laws~\cite{Stryker2019}, a recent work by some of us described a complete and fault-tolerant error correction scheme for $\mathbb{Z}_2$ lattice gauge theories in one and two spatial dimensions~\cite{Rajput2023}. In this work we extend the QEC scheme presented there in several ways. Our first new contribution is a generalization of the Gauss' Law code of Ref.~\cite{Rajput2023} to arbitrary spatial dimensions and at the same time improving the resource requirements of the original proposals for $2D$. We then characterize the logical operations of the code and show how to leverage the error correction formulation to express the logical Hamiltonian operator in the physical subspace effectively removing the gauge constraint. In the models we consider here, the resulting logical Hamiltonian can be thought of as being obtained integrating out the fermionic degrees of freedom living on the lattice sites. The idea itself is not new and several explicit constructions for mapping a lattice gauge into a spin model by removing the fermionic matter have been proposed (see e.g.~\cite{PhysRevB.98.075119,PhysRevLett.124.120503,PhysRevB.105.075132}). Typically these strategies are tailored to specific spatial dimensions and require representing the fermions in terms of Majorana modes, our approach on the other hand is independent on the spatial dimension and easy to construct. Finally, the original proposal for these Gauss' Law codes in Ref.~\cite{Rajput2023} focused solely on the use of the error-correcting code to store information in a fault-tolerant way, in this work we show how to implement  universal gate set on a Gauss' Law code and provide a complete fault-tolerant algorithm to perform Hamiltonian simulation in this framework. For simulations using product formulas we present a scheme that requires only 3 (logical) ancilla qubits in addition to the memory qubits required to encode the system.

In Section~\ref{section:2}, we define the Hamiltonian and the LGT we take into account, as well as the corresponding gauge operators. In Section~\ref{section:3}, we review the basic concepts behind the stabilizer formalism and we will show how to exploit the gauge symmetry to do error correction. Moreover, we will explicitly build a distance-3 error-correcting code. In Section~\ref{section:4} the Hamiltonian is written in terms of the logical operations of the error-correcting code, reducing the total number of degrees of freedom to those required to span the physical Hilbert space. We will also show that the new Hamiltonian can be written as a function of only hard-core bosonic degrees of freedom. Finally, in Section~\ref{section:5}, we discuss how simulation algorithms can be implemented using Quantum Signal Processing~\cite{PhysRevLett.118.010501,Low2019hamiltonian} and Trotterization~\cite{childs2021theory} in a completely fault tolerant way using the error-correcting code developed in Section~\ref{section:3}. This not only proves that the Gauss' Law can be exploited to correct errors, but also shows how to use it to perform fault-tolerant simulations.

\section{Abelian lattice gauge theories and Gauss' Law}
\label{section:2}
We will consider a $\mathbb{Z}_N$ model Hamiltonian which contains the fermionic mass and hopping terms, the electric field term, and the self-interaction of the gauge field, which translates into the plaquette interaction \cite{Horn1979,Wiese_2013,Zohar2016}. The theory is defined on a cubic lattice in $d$ spatial dimensions with fermions defined on the lattice sites and gauge bosons on the links. We label the sites of the lattice as $S_{\vec{l}}$, where $\vec{l}$ are $d$-dimensional vectors with integer components, and define the set of $d$ unit vectors in the positive directions as $\mathcal{D}$ (e.g., for $d=2$ we have $\mathcal{D}=\{(1,0),(0,1)\}$. With this notation, the Hamiltonian of the system can be written as:
\begin{equation}
    H = H_{M} + H_{hop} + H_{E} + H_{P}
    \label{Full_Hamiltonian}
\end{equation}
The different terms correspond, in order, to a mass term $H_M$, an hopping term $H_{hop}$, an electric term $H_E$ and a plaquette, or magnetic, term $H_P$. The explicit definition of the mass term is as follows
\begin{equation}
H_{M} = m \sum_{\vec{l}} \sigma_{\vec{l}}\; \psi_{\vec{l}}^{\dagger}\psi_{\vec{l}}\;,
\end{equation}
with $\psi_{\vec{l}}$ ($\psi^{\dagger}_{\vec{l}}$) fermionic annihilation (creation) operators acting on the sites $S_{\vec{l}}$ of the lattice and 
\begin{equation}
\sigma_{\vec{l}} = (-1)^{l_1+\dots+l_d}:= (-1)^{|l|}\;,    
\end{equation}
a sign factor. The hopping term is instead given by
\begin{equation}
H_{hop} = \epsilon \sum_{\vec{l}}\sum_{\hat k\in \mathcal{D}} \sigma_{\vec{l},\hat k} \left(\psi_{\vec{l}}^{\dagger} Q_{\vec{l},\hat k} \psi_{\vec{l}+\hat k} + \psi_{\vec{l}+\hat k}^{\dagger} Q^{\dagger}_{\vec{l},\hat k} \psi_{\vec{l}}\right)\;.
\end{equation}
The operator $Q_{\vec{l},\hat k}$ acts on the link $L_{\vec{l},\hat k}$ which connects the site $S_{\vec{l}}$ to the site $S_{\vec{l}+\hat k}$ as a lowering operator in the computational basis
\begin{equation}
\label{eq:Q}
Q \ket{m} = \ket{m-1}\;,
\end{equation}
while $\sigma_{\vec{l},\hat k}$ is a sign factor taking values $(-1)$ in the $1$-direction, $(-1)^{l_1}$ in the $2$ direction and $(-1)^{l_1+\dots+l_{d-1}}$ in the $d$ direction. The electric term can be written as
\begin{equation}
H_{E} = \lambda_E \sum_{\vec{l}} \sum_{\hat k\in \mathcal{D}}\left( P_{\vec{l},\hat k} + P_{\vec{l},\hat k}^{\dagger} \right) \;,
\end{equation}
with the operators $P_{\vec{l},\hat k}$ acting on the links $L_{\vec{l},\hat k}$ as
\begin{equation}
\label{eq:P}
P\ket{m}=e^{i\frac{2\pi}{N}m}\ket{m}\;,
\end{equation}
in the computational basis. The $P$ and $Q$ operators satisfy the $\mathbb{Z}_N$ algebra
\begin{equation}
P^N = Q^N = 1 \quad\quad\quad P^{\dagger} Q P = e^{i\frac{2\pi}{N}} Q\;.
\label{eq:zn_algebra}
\end{equation}
In a $U(1)$ theory, the electric term is a sum of the Casimir operators $E^2_{\vec{l},\hat k}$. The connection with the $Z_n$ case can be done by writing the electric term as
\begin{equation}
    H_E =\lambda_E\sum_{\vec{l}}\sum_{\hat k \in \mathcal{D}} O_{\vec{l},\hat k}
\end{equation}
where the operator $O_{\vec{l},\hat k}$ is in general
\begin{equation}
    O = \sum_{j}f(j)\ket{j}\bra{j}
\end{equation}
The only constraint is on $f(j)$, which as to approach $f(j)\rightarrow j^2$ as $N\rightarrow \infty$. So we can chose
\begin{equation}
    f(j) = 4\sin^2\left( \frac{\pi j}{N} \right)
\end{equation}
and expanding in complex exponentials we get
\begin{equation}
    \sum_j f(j)\ket{j}\bra{j} = \sum_j \left( e^{i\frac{2\pi}{N}j} + e^{-i\frac{2\pi}{N}j} \right) \ket{j}\bra{j} = P + P^\dag
\end{equation}

Finally, the plaquette term takes the form
\begin{equation}
H_{P} = \lambda_P \sum_{p} \left(W_{p} + W^{\dagger}_{p}\right)\;,
\end{equation}
where the sum is over all the plaquettes in the lattice and the operators $W_p=Q_1Q_2Q^\dagger_3Q^\dagger_4$ are formed by four $Q$ operators acting on the links of the plaquette using the ordering convention shown in Fig.~\ref{fig:plaquetteConvention}.
\begin{figure}[t]
    \centering
    \includegraphics[scale=0.3]{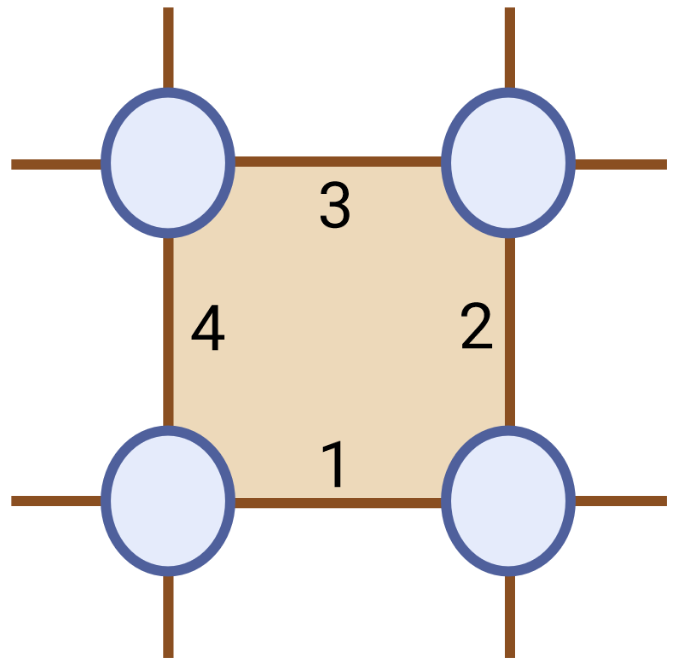}
    \caption{Labelling convention for links in a plaquette.}
    \label{fig:plaquetteConvention}
\end{figure}

\begin{figure}[b]
    \centering
    \includegraphics[scale=0.3]{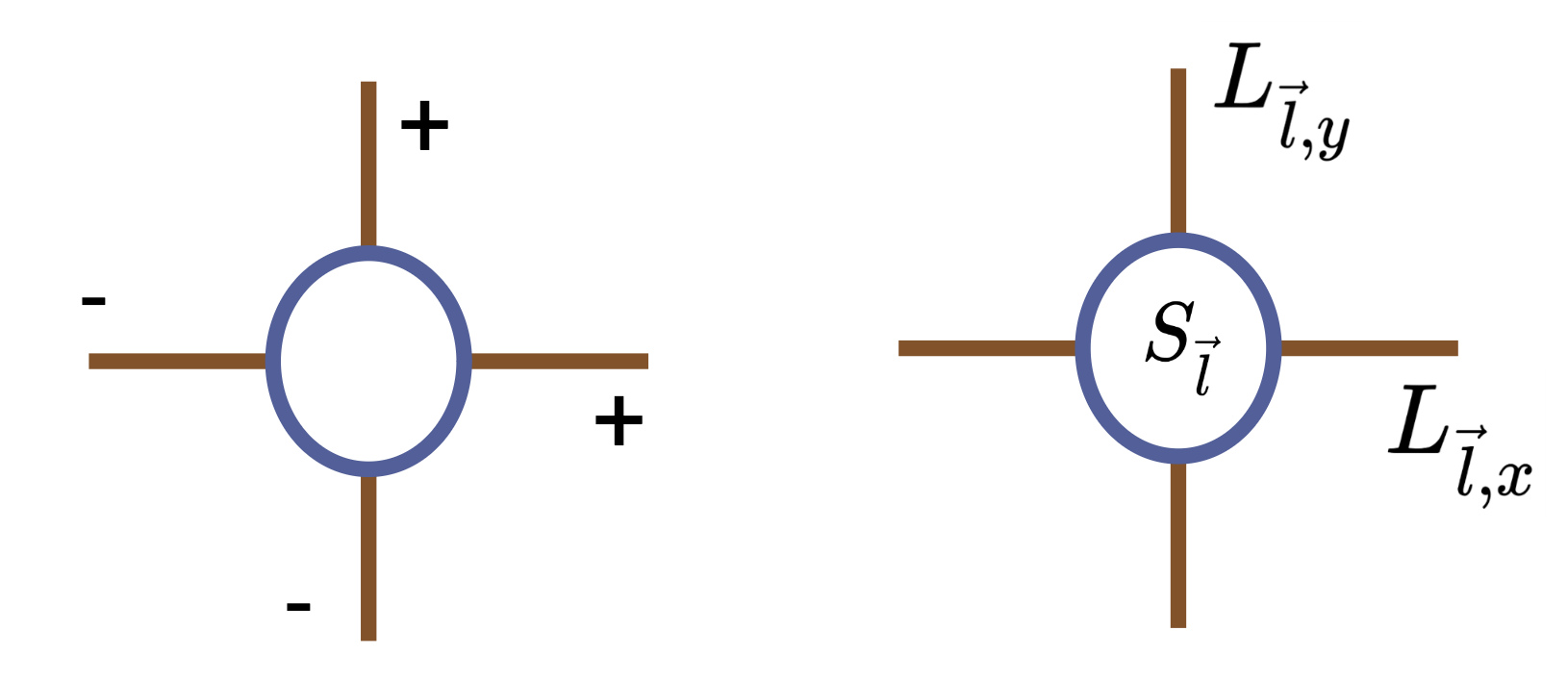}
    \caption{Convention for positive and negative links around a site on the left-hand side, and labelling convention on the right-hand side.}
    \label{fig:linkConvention}
\end{figure}

In terms of the basic building blocks defined above, the Gauss' Law operator can be written as
\begin{equation}
    G_{\vec{l}} = q_{\vec{l}}^{\dagger} \prod_{\hat k\in \mathcal{D}} P^{\dagger}_{\vec{l},\hat k} P_{\vec{l}-\hat k,\hat k}\;,
    \label{GL_theoretical}
\end{equation}
where $q_{\vec{l}}$ is a unitary operator connected with the charge $n_{\vec{l}}$ at a site
\begin{equation}
q_{\vec{l}} = e^{-i\frac{2\pi}{N}n_{\vec{l}}}\quad\quad n_{\vec{l}}=\psi_{\vec{l}}^\dagger\psi_{\vec{l}}-\frac{1}{2}\left( 1- (-1)^{|l|} \right)\;.
\end{equation}
The term $\psi_{\vec{l}}^\dagger\psi_{\vec{l}}$ in $n_{\vec{l}}$ gives $0$ if an even site is empty or an odd site is full, and $1$ in the other cases. The convention used in the definition of Eq.~\eqref{GL_theoretical} is to consider the incoming links into a site to be negative and the outgoing as positive (see Fig.~\ref{fig:linkConvention}). Physical states of the theory, i.e. those that satisfy the local $\mathbb{Z}_N$ symmetry, are thus eigenstates of the operators $G_{\vec{l}}$ at every site of the lattice
\begin{equation}
G_{\vec{l}}\ket{phys}=e^{-i\frac{2\pi}{N}\epsilon_{\vec{l}}}\ket{phys}\quad\forall \vec{l}\;,
\end{equation}
where the $\epsilon_{\vec{l}}$ denote the possible static charge present on the site $S_{\vec{l}}$. Different assignments of these charges define sectors of the physical theory which do not mix between themselves. For simplicity in the rest of the present work we will work in the sector with zero static charges but, if needed, other sectors can be described with the same techniques. In the zero charge sector, physical states obey therefore
\begin{equation}
G_{\vec{l}}\ket{phys}=\ket{phys}\quad\forall \vec{l}\;.
\label{eq:GaussLawPhys}
\end{equation}
This relation can be connected to Gauss' Law by first introducing electric field operator $E_{\vec{l},\hat k}$ as 
\begin{equation}
P^{\dagger}_{\vec{l},\hat k}=e^{i\frac{2\pi}{N}E_{\vec{l},\hat k}}\;,
\end{equation}
Then, Eq.~\ref{eq:GaussLawPhys} becomes:
\begin{equation}
    e^{i\frac{2\pi}{N}\left(\div E_{\vec{l}} - n_{\vec{l}}\right)}\ket{phys} = \ket{phys}
\end{equation}
and it is equivalent to the following condition
\begin{equation}
    \left[ \div E_{\vec{l}} - n_{\vec{l}} \right]_{\hspace{-10pt} \mod N} \ket{phys} = \left[ \sum_{\hat k \in \mathcal{D}} \left(E_{\vec{l}, \hat k} - E_{\vec{l}-\hat k, \hat k} \right) - n_{\vec{l}} \right]_{\hspace{-10pt} \mod N} \ket{phys} = 0
    \label{eq:GaussLawNonUnitary}
\end{equation}
Since Eq.~\eqref{eq:GaussLawPhys} and Eq.~\eqref{eq:GaussLawNonUnitary} are equivalent, we will use the condition in Eq.~\eqref{eq:GaussLawPhys} since it is unitary.

As it is done in \cite{Rajput2023}, we will limit ourselves to the choice of the electric field cutoff equal to $2$, which means that the symmetry group becomes $\mathbb{Z}_2$.  It further allows us to implement the stabilizer operations using low-weight Pauli operations and also allows us to represent the field via a single physical qubit.  In this case the $P$ and $Q$ operators from Eq.~\eqref{eq:P} and Eq.~\eqref{eq:Q}, and the Gauss law operator, can be chosen as
\begin{equation}
    P = Z, \quad Q = X, \quad G_{\vec{l}} = (-1)^{|l|} Z_{S_{\vec{l}}} \prod_{\hat k\in \mathcal{D}} Z_{L_{{\vec{l}},\hat k}}Z_{L_{{\vec{l}}-\hat k,\hat k}}
    \label{pauliZ_representation_PQG}
\end{equation}
where $X,Z$ are the standard Pauli matrices $\sigma_X, \sigma_Z$. In this case of $\mathbb{Z}_2$, those operators $P,Q$, and $G$ are all Pauli operators, which means that we can interpret the set of $G$ as the generator of the stabilizer group. Of course, it is not enough since it commutes with every $Z$ error, but it is not a problem since we can always concatenate two classical error-correcting codes to get a quantum one.

\section{Quantum error correction with Gauss' Law}
\label{section:3}

In this section we introduce the error-correcting code that uses Gauss' law as generators of the stabilizer group generalizing the results derived in Ref.~\cite{Rajput2023}. First, let us recall how stabilizers work: we need an abelian subgroup of the Pauli group that will be the stabilizer group. The generators of this group are then called stabilizers, while the codewords (states protected against errors by the code) are defined as the eigenvectors with eigenvalue $+1$ of every stabilizer operator. By measuring the eigenvalue of those stabilizers (which means extracting the error syndrome) we can understand if an error occurred with the further intention of deducing which qubit the error happened on.  Of course, all errors cannot be detected in this manner: only errors that anti-commute with at least one stabilizer $S_i$ can be measured in this fashion since it will change the eigenvalue of that stabilizer from $+1$ to $-1$ when measured.  Errors that commute with the stabilizer group lead to logical errors, which cannot be corrected.

Now, we will prove that the group generated by the Gauss' Law operators is a stabilizer group of an error-correcting code able to correct every single-qubit $X$ error, which is the same of saying it generates a classical distance-$3$ error-correcting code. 

\begin{theorem}
    Let \(H\) be the Hamiltonian of Eq~\eqref{Full_Hamiltonian}, with a \(\mathbb{Z}_2\) gauge group, and let \(G_{\vec{l}}\) of Eq.~\eqref{GL_theoretical} be the generators of the local symmetry. Given a \(d\)-dimensional lattice with \(N\) sites and \(dN\) links, mapped into \(N+dN\) qubits, a distance $3$ quantum error-correcting code that corrects single-qubit $X$ errors (a classical distance-\(3\) error-correcting code) can be built using only \(G_{\vec{l}}\) as stabilizers for every $N\ge 3^d$.
    \label{theorem_gauss_law}
\end{theorem}
\begin{proof}
For simplicity consider the $1$-dimensional system depicted in Fig.~$\ref{fig:1Dsystem}$.
\begin{figure}[t]
    \centering
    \includegraphics[scale=0.4]{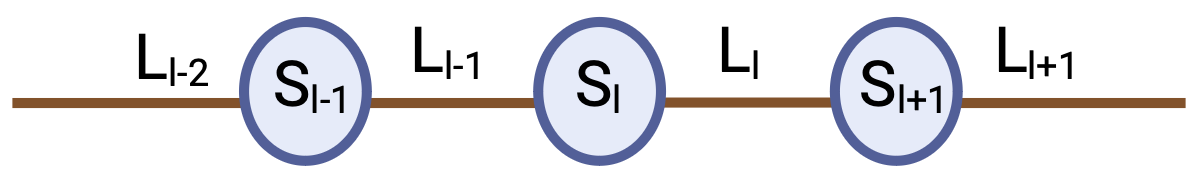}
    \caption{Representation of a 1-dimensional lattice, showing the labelling convention for sites and links.}
    \label{fig:1Dsystem}
\end{figure}
Since the Gauss' law operator on the $l$-th site is $G_l = (-)^l Z_{L_{l-1}} Z_{S_l} Z_{L_{l}}$, it commutes with every $Z$ error, so the stabilizer group generated by the Gauss' law will never be able to correct $Z$ errors. However, if an $X$ error happens on the site $S_l$, the operator $G_l$ will give $-1$ eigenvalue, while if an $X$ error happens on the link $L_l$, both $G_l$ and $G_{l+1}$ will give $-1$ eigenvalue. Indeed, by measuring $G_{l-1}, G_l$, and $G_{l+1}$ we are able to correct every bit-flip error on $L_{l-1}, S_l$, or $L_{l}$, as shown in Table \ref{tab:error_location_1D_single_site}. This is an error-correcting code that, acting on $N$ sites and $N$ links (where every link and every site are one qubit each), encodes $N$ logical qubits. We can say that it is a
\begin{equation}
    [2N, N, 3]
\end{equation}
classical error-correcting code, with $2N$ physical qubits, $N$ stabilizers and so $2N-N$ logical qubits. Since we just proved that it can catch every single-qubit $X$ error, it has distance $3$.

\begin{table}[t]
    \centering
    \begin{tabular}{|c|c|c|c|}
        \hline
        $G_{l-1}$ & $G_l$ & $G_{l+1}$ & error location \\
        \hline
        $+$ & $+$ & $+$ & none \\
        \hline
        $-$ & $-$ & $+$ & $L_{l-1}$ \\
        \hline
        $+$ & $-$ & $+$ & $S_l$ \\
        \hline
        $+$ & $-$ & $-$ & $L_{l}$ \\
        \hline
    \end{tabular}
    \caption{In this table we reported only the errors on $S_l$ and the links attached to it, so where the $G_l$ fails. Other outcomes correspond to multiple errors that we cannot correct, or to errors in $S_{l\pm 1}$ for which we need to look at $G_{l\pm 2}$.}
    \label{tab:error_location_1D_single_site}
\end{table}

In the $2$ dimensional case, looking at Eq.~\eqref{GL_theoretical}, the Gauss' law becomes a weight $5$ operator acting on $1$ site and $4$ links. Again, if only one Gauss' law gives $-1$ means that the error is on the site, while if two neighbouring Gauss' laws give $-1$ the error is on the link they share. It is easy to generalize both the Gauss' law and the decoding to $d$ dimensions as follows: the Gauss' law will always act on $1$ site and all neighbouring links, and if only one Gauss' law gives $-1$ the error is on the site, while if two neighbouring Gauss' law give $-1$ the error is on the link they share. 

As for the parameters of the code, let us remind that in a code with $n$ physical qubits and $n-k$ stabilizers we have $k$ logical qubits. Since we will always have $dN$ links and $N$ sites, we will have a total of $N+dN$ physical qubits and $N$ stabilizers, which leads to $dN$ logical qubits and the code can be written as
\begin{equation}
    [N+dN, dN, 3]
\end{equation}
We assume periodic boundary conditions, and so we need for this code the neighbouring sites we are considering to be different. In the $d=1$ case, every site has $2$ neighbours, which means that we need at least $N=3$. For a general $d$, we need a grid of $3^d$ sites at least.
\end{proof}

We proved that the redundancy given by the Gauge symmetry is enough to build an error-correcting code able to correct one type of errors. Now, we will use it to build two different types of distance-$3$ quantum error-correcting codes by using concatenation.

\begin{lemma}
    Consider the classical distance-3 error-correcting code of Theorem \ref{theorem_gauss_law}, built on a $d$-dimensional lattice with \(N+dN)\) qubits and that uses the Gauss' law operators $G_l$ as stabilizers. That code has parameters \( [N+dN, dN, 3] \), \(\forall N\ge 3^d\). Using concatenation, it is possible to build a CSS quantum error-correcting code with parameters 
    \begin{equation}
    [[N+dN+O(\log(N+dN)), dN, 3]] \quad \forall N\ge3^d.  
    \end{equation}
    \label{lemma:asymptotic_optimality}
\end{lemma}
The proof of Lemma \ref{lemma:asymptotic_optimality} is given in Appendix \ref{Other_phase_flip}. We now focus on a class of error-correcting codes that are not as optimal in terms of memory, but have improved locality for the stabilizer measurements. Let us define the quantum low-density parity check (QLDPC) codes as codes in which the number of qubits involved in each check and the number of checks acting on each qubit are both upper bounded by a constant \cite{PRXQuantum.2.040101}.  We can then use this definition to see that the Gauss' law error-correcting code is then a CSS QLDPC code.  We provide a formal statement of this observation while providing the parameters for the code below.

\begin{lemma}
    Consider the classical distance-3 error-correcting code of Theorem \ref{theorem_gauss_law}, built on a $d$-dimensional lattice with \(N+dN\) qubits, and that uses the Gauss' law operators as stabilizers. This code has parameters \( [N+dN, Nd, 3] \), for \(N>3^d\). Using concatenation, it is possible to build a CSS QLDPC code with parameters 
    \begin{equation}
        [[3(N+dN), Nd, 3]] \quad \forall N>3^d
    \end{equation}
    \label{lemma:QLDPC}
\end{lemma}

\begin{proof}
From Theorem~\ref{theorem_gauss_law}, the error-correcting code built from the Gauss' law operators can detect and correct single-qubit bit-flip error ($X$ errors). To transform it into a distance-3 quantum error-correcting code, it has to be concatenated with another error-correcting code able to correct every single-qubit phase-flip error ($Z$ errors). The simplest and smallest phase-flip code is the $3$ qubit repetition code, which takes $3$ physical qubits and encodes them into one logical qubit. The stabilizers for this code are $X_1X_2$ and $X_2X_3$ and we have two possibilities to concatenate it to the Gauss' law code. One possibility is to follow Ref.~\cite{Rajput2023} and transform every site and link of Fig.~\ref{fig:1Dsystem} into logical qubits so that every link and every site will correspond to $3$ physical qubits. Then, those logical qubits already protected against phase-flip errors will be used as physical qubits for the bit-flip Gauss' law code. To do that, we have to substitute the $Z$ operations in the Gauss' law, with the $\bar Z$ logical operation on the logical qubit of the phase-flip code. Since the logical operations for the repetition code are $\overline{X} = X_i, \overline{Z} = Z_1Z_2Z_3$, the Gauss' operator $G_l$ becomes a weight $9$ operator. A representation of the qubits layout and of this code are shown in Fig.~\ref{fig:representation_concatenation} and Fig.~\ref{fig:representation_1d_1concatenation} respectively.
\begin{figure}[t]
    \centering
    \includegraphics[scale=0.4]{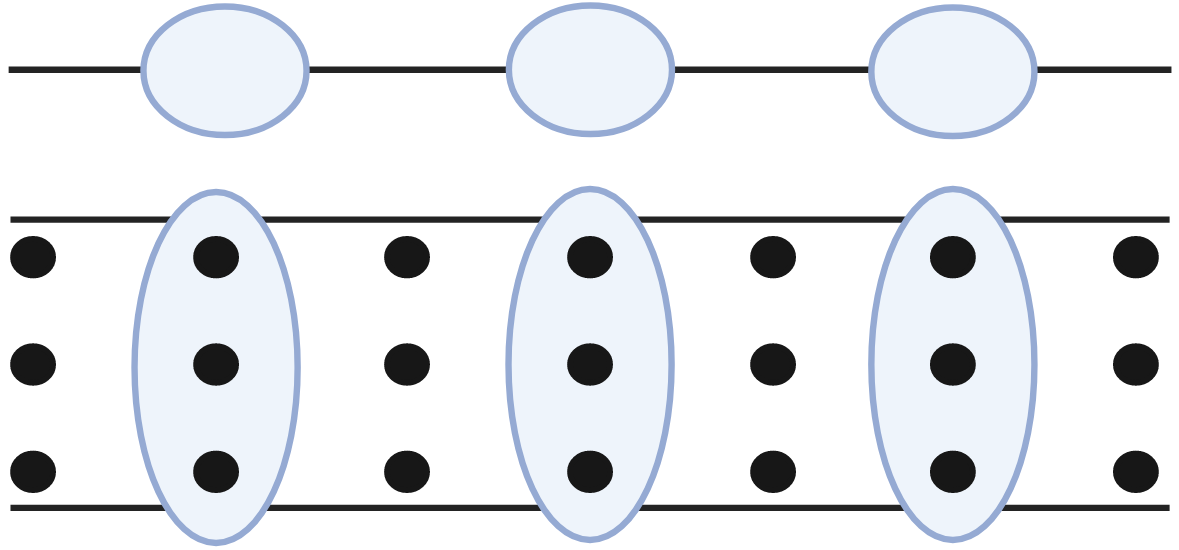}
    \caption{Graphical representation of the additional qubits needed for the concatenation. This figure shows how we can associate every site and link to $3$ qubits now, instead of being one qubit each as before.}
    \label{fig:representation_concatenation}
\end{figure}
\begin{figure}[b]
    \centering
    \includegraphics[scale=0.4]{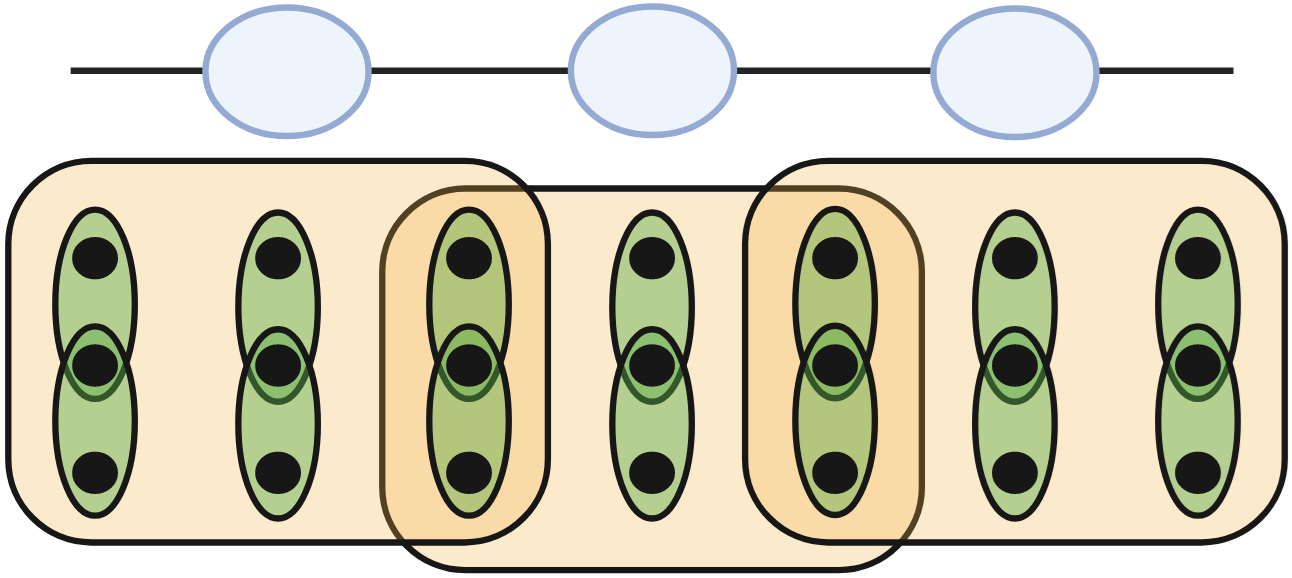}
    \caption{Graphical representation of the quantum error-correcting code. The green circle represents $X$ stabilizers, which are all of weight $2$. The orange area instead represents the weight $9$ Gauss' law stabilizer. On top, the 1-dimensional system is represented to show explicitly that every link and site is in reality made of the $3$ physical qubits below it.}
    \label{fig:representation_1d_1concatenation}
\end{figure}

The second possibility is the other order of concatenation, by building the Gauss' law code before and then repeating it to build the repetition code on top. In this case, we will have $3$ copies of the Gauss' law codes, with logical operations $\overline{Z}_l = Z_{L_l}$ and $\overline{X}_l = X_{S_l}X_{L_{l}}X_{S_{l+1}}$. Then we have to take the stabilizers of the repetition code and extend them from $X_1X_2$ to $\overline{X}_1\overline{X}_2$ and so on, as we did for the Gauss' law in the previous concatenation. A graphical representation of this code is shown in Fig.~\ref{fig:representation_1d_2concatenation}.
\begin{figure}[h]
    \centering
    \includegraphics[scale=0.4]{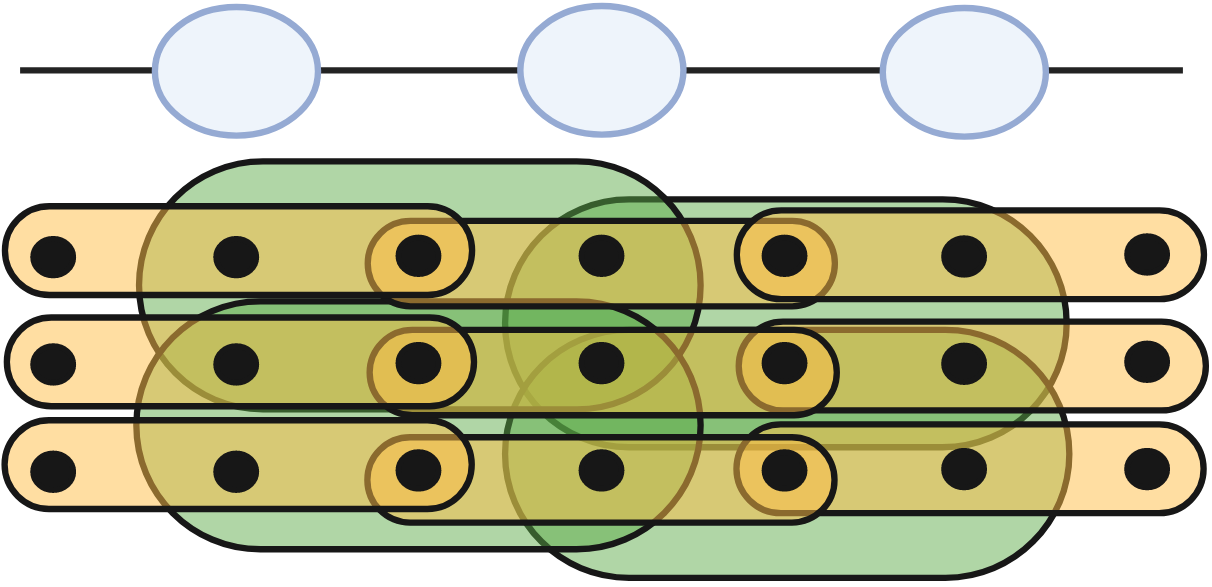}
    \caption{Graphical representation of the quantum error-correcting code with the second possible concatenation. The green areas represent the weight $6$ $X$ stabilizers of the phase-flip repetition code. The orange areas instead represent the weight $3$ Gauss' law stabilizers. On top, the 1-dimensional system is represented to show explicitly that every link and site is in reality made of the $3$ physical qubits below it.}
    \label{fig:representation_1d_2concatenation}
\end{figure}
Both concatenations have the same number of qubits ($3$ for every link and site in the system) and the same logical operations which are $\overline{X}_l = X^i_{S_l} X^j_{L_l} X^k_{S_{l+1}}$ with $i,j,k \in \{1,2,3\}$ label the $3$ qubits of a link/site, and $\overline{Z}_l = Z^1_{L_l} Z^2_{L_l} Z^3_{L_l}$. The major difference between the two concatenations is the maximum weight of the stabilizers ($9$ and $6$) and the fact that, with the first concatenation, we were able to correct $1$ phase-error for every link and every site independently, while the second concatenation requires only $1$ phase-flip error in a larger patch of the system. However, while the first concatenation has a very large patch in which only one bit-flip error can occur, the second concatenation has more symmetry between the errors. So one could decide which concatenation to use according to the error model of the hardware being used. Both concatenations give a distance $3$ quantum error-correcting code, and in both cases the maximum weight (and the locality) of stabilizers is bounded.

In the $2$-dimensional case nothing changes. The Gauss' law acts on a site and $4$ links instead of $2$, but the code construction remains the same because we can still use the $3$ qubits repetition code and concatenate it to the Gauss' law code. Even if the Gauss' law operator changes slightly, we can still encode every link and site in the repetition code and use those logical qubits for the Gauss code, or we can make $3$ copies of the Gauss code and use on top of that the repetition code. This time, the maximum weight of operators for the two concatenations is respectively $15$ and $6$. The $2$ dimensional system is shown in Fig.~\ref{fig:representation_2d}.

\begin{figure}[h]
    \centering
    \includegraphics[scale=0.6]{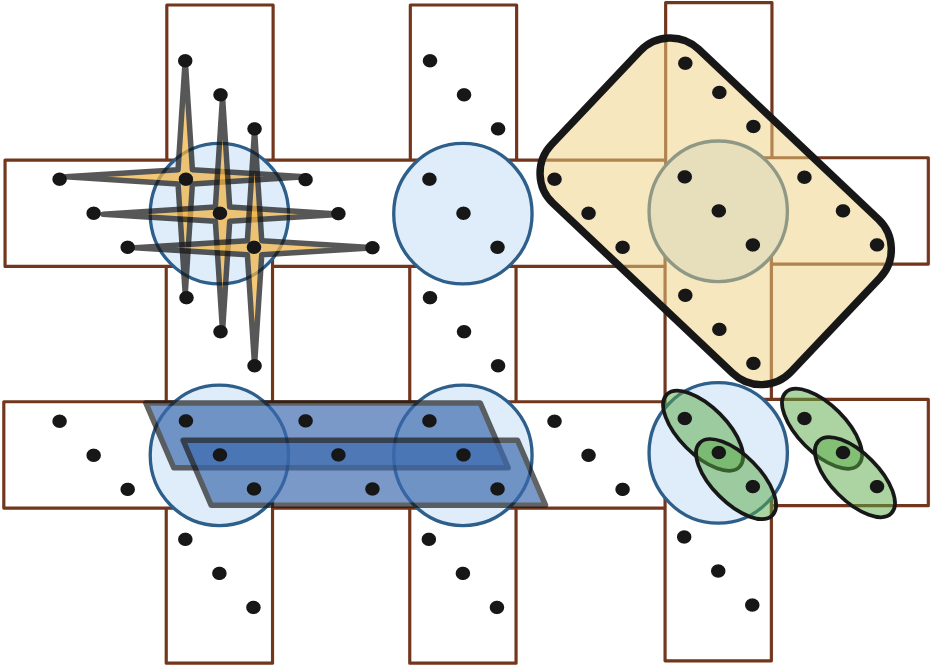}
    \caption{Graphical representation of the stabilizers and logical operations on a $2$-dimensional lattice. The first concatenation leads to the orange rectangles, which is the Gauss' law (weight $15$ operator), and the green ellipses that are the phase-flip stabilizers. The second concatenation instead leads to $3$ Gauss' law per site, depicted as the yellow stars, and the blue parallelograms that are the phase-flip weight $6$ stabilizers.}
    \label{fig:representation_2d}
\end{figure}

For a generic value of $d$, we already discussed how to generalize the classical Gauss' law code with Theorem \ref{theorem_gauss_law}. To generalize the full quantum error-correcting code we have to describe how the $3$-qubit phase-flip code behaves. This depends on the choosen concatenation: if we build the phase-flip as first code, we encode every link and site in a logical qubit already protected against phase-flip errors, and we only need to build on top of that the Gauss' law code as described in Theorem \ref{theorem_gauss_law}. If instead we build the Gauss' law code first, we need, as before, $3$ copies of it, and as one can see from the $2$-dimensional case we will have $2$ phase-flip stabilizers (that are weight $6$ operators) on every link. For completeness, let us report the maximum weight operators for the different concatenations: for the first concatenation, the operator with higher weight will always be the Gauss' law with weight $3(2d+1)$ (corresponding to $3$ qubit for the site and $3$ qubits for every link which are $2$ for each dimension). For the second concatenation, the weight of the Gauss' law operator will be $2d+1$ (with only one qubit for the site and every link), and the maximum weight will be the max between $2d+1$ and $6$ which is the weight of the phase-flip stabilizers.
\end{proof}

In the $2$-dimensional case we need to measure $5$ Gauss' law operators to fully correct a site and its $4$ links. This means that, for our code to work, we have to assume that only $1$ error happens at a time in those $5$ sites and the corresponding $16$ links. This is a big assumption because we would need a quite low error rate, but we can relax the assumption. Once we measure a Gauss' law and it gives $-1$, to know whether the error is on a specific link we need to measure a neighboring Gauss' law extending the measured system to $1$ site and other $3$ links. To avoid it, we can add $3$ qubits, corrected with a phase-flip repetition code, and add the $ZZ$ stabilizer that measures the parity between this additional logical qubit and the already existing link. This procedure is shown in Fig.~\ref{fig:link_doubling}, underlying the qubits that an additional Gauss' law would bring into the considered patch, compared to the single additional link. Of course, this would require more qubits for the error-correcting procedure, but if the assumption is that we can have only $1$ error at a time for every $K$ qubits, we can adjust $K$ according to the error rate in this way. If we add no qubits, we ask for only $1$ error in 63 qubits, corresponding to $5$ sites and their $16$ links (where every site and link corresponds to $3$ qubits), while if we double every link, we are asking for only $1$ error in $4$ doubled links and $1$ site for a total $27$ physical qubits.

\begin{figure}[t]
    \centering
    \includegraphics[scale=0.4]{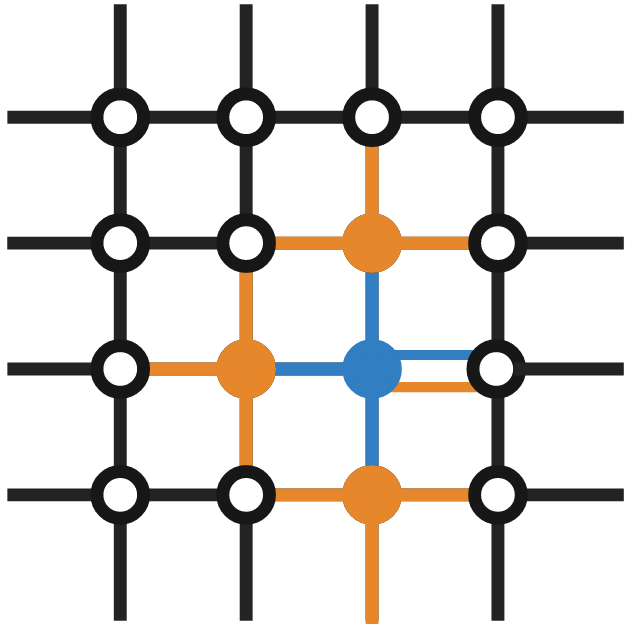}
    \caption{Graphical representation of the patch of the system we need to measure to correct $1$ site and its $4$ links. We start by measuring the Gauss' law on that site, and then we have to check if the error is on the site or a link. To do that, for every link we can either measure another Gauss' law, adding a site and $3$ links ($12$ qubits in total), or add $3$ qubits and measure the parity between the existing link and the additional one.}
    \label{fig:link_doubling}
\end{figure}

Fault-tolerant measurement of error syndromes can be implemented using the ``logical to physical CNOT" explained in \cite{Rajput2023}, and one flag qubit.  With this observation in place, we see that we can perform a fault-tolerant implementation of the syndrome measurements.  Further, the fact that the code is a CSS code will be invaluable later when we discuss how to simulate dynamics within these codes on a quantum computer.

\section{Hamiltonian in terms of logical operations}
\label{section:4}
Above we showed how we can construct quantum error-correcting codes using the Gauss' law checks.  However, in order to simulate lattice Gauge theories within the quantum error-correcting code we need to provide a representation of the Hamiltonian expressed in terms of the corresponding logical operations. In this section, we will rewrite the Hamiltonian in terms of the logical operations of the Gauss' law error-correcting code described in Theorem~\ref{theorem_gauss_law}. Then, we will be able to define new creation and annihilation operators and write the Hamiltonian in terms of these new degrees of freedom. This will be equivalent to integrating fermions on sites, and we will see that the new degrees of freedom will be hardcore bosons.

So, let us start by writing the Hamiltonian in terms of the logical operations, reducing the degrees of freedom.

\begin{theorem}
    Let \(H\) be the Hamiltonian of Eq.~\eqref{Full_Hamiltonian} with the gauge group \(\mathbb{Z}_2\). It is always possible to write $H$ in terms of the logical operations of the error-correcting code of Theorem~\ref{theorem_gauss_law}. By doing so, the degrees of freedom in the Hamiltonian are reduced from $N+dN$ to $dN$.
    \label{theorem:H_logical_operations}
\end{theorem}

Here we will discuss only the results for the $1$ dimensional case, while the proof of Theorem~\ref{theorem:H_logical_operations}, as well as the complete calculation for the $1$ and $2$ dimensional cases, are reported in Appendix~\ref{bosons_2D}.

Recall the $1$ dimensional Hamiltonian of Eq.~\eqref{Full_Hamiltonian}
\begin{equation}
    H = m \sum_l (-1)^l \psi_l^{\dagger}\psi_l + \epsilon \sum_{l} (\psi_l^{\dagger} Q_{l} \psi_{l+1} + \psi_{l+1}^{\dagger} Q_{l}^\dag \psi_{l}) +2\lambda_E \sum_{l} P_{l}
\end{equation}
With a $\mathbb{Z}_2$ symmetry we already said in the first chapter that $P_{L_l} = Z_{L_l}$ and $Q_{L_l} = X_{L_l}$. As for the fermionic operators, we can use the Jordan-Wigner representation, so that we can write the Hamiltonian in terms of the operations $Z$ and $X$ on links and sites:
\begin{equation}
\begin{split}
    H &= \frac{m}{2}\sum_{l} (-1)^l(1 -(-1)^l Z_{S_l}) +\frac{\epsilon}{2} \sum_l \left( 1 + Z_{S_l}Z_{S_{l+1}} \right) X_{S_l}X_{S_{l+1}}X_{L_{l}} +2\lambda_E \sum_l Z_{L_l}
\end{split}
\end{equation}
Remember that, due to the staggered fermions representation, the operator $\rvert0\rangle\langle1\lvert$ is the annihilation operator for even sites while it is the creation operator for odd sites. Now, recall the logical operations of the error-correcting code of the last chapter and the Gauss' law operators:
\begin{equation}
    \overline{X}_l = X_{S_{l}}X_{L_{l}}X_{S_{l+1}} \quad \overline{Z}_l = Z_{L_l} \quad G_l = (-1)^l Z_{L_{l-1}}Z_{S_l}Z_{l}
    \label{eq:logic}
\end{equation}
Here we are neglecting the presence of the phase-flip code, but the only difference with it is that every $Z$ operation would be a weight $3$ logical $\bar Z$. From here, using the relations of Eq.~\eqref{eq:logic}, we can express the Hamiltonian in terms of logical operations only:
\begin{equation}
\begin{split}
    H &= \frac{m}{2}\sum_{l} (-1)^l(1 - G_l\overline{Z}_{l-1} \overline{Z}_{l}) +\frac{\epsilon}{2} \sum_l \left( 1 -  G_lG_{l+1}\overline{Z}_{l-1} \overline{Z}_{l+1} \right) \overline{X}_{l} +2\lambda_E \sum_l \overline{Z}_l\\
    &= \frac{m}{2}\sum_{l} (-1)^l(1 - \overline{Z}_{l-1} \overline{Z}_{l}) +\frac{\epsilon}{2} \sum_l \left( 1 -  \overline{Z}_{l-1} \overline{Z}_{l+1} \right) \overline{X}_{l} +2\lambda_E \sum_l \overline{Z}_l\\
\end{split}
    \label{eq:1d-hamiltonian}
\end{equation}
Note that in the second line the Gauss' laws have been simplified since they correspond to the logical identity. This form of the Hamiltonian is the one that Theorem~\ref{theorem:H_logical_operations} states to exist and the degrees of freedom have been reduced from $2N$ to $N$ as expected. Apart from the explicit staggered mass term, this manifestly gauge invariant form of the Hamiltonian for a one-dimensional $\mathbb{Z}_2$ theory coupled to spinless fermions has been proposed already in Ref.~\cite{PhysRevLett.124.120503} where it was derived a Majorana representation for the fermions. Our construction is however simpler and allows for straightforward generalizations to higher spatial dimensions.

Now that we have the Hamiltonian in term of the logical degrees of freedom, we can introduce bosonic degrees of freedom and write the entire Hamiltonian in terms of those.

\begin{corollary}\label{cor:bosonic}
    Let \(H\) be the Hamiltonian of Eq.~\ref{Full_Hamiltonian} with the gauge group \(\mathbb{Z}_2\). It is possible to write $H$ in terms of only bosonic degrees of freedom.
\end{corollary}

\begin{proof}
    Let us start from the Hamiltonian in the form given by Theorem \ref{theorem:H_logical_operations}, such that the Hamiltonian $H$ is only a function of $X_i^k$ and $Z_i^k$ ($i\in [1,N]$ labels sites in the lattice, while $k\in [1,d]$ labels the different directions in the lattice). In general, the Hamiltonian will be written in terms of $2dN$ operators, that correspond to $dN$ degrees of freedom (for a $d$ dimensional lattice with $N$ sites). Now, let us define a set of bosonic degrees of freedom as follows:
    \begin{equation}
    \begin{split}
        \phi^k_i &= \frac{1}{2}\left( 1 - Z_i^k \right) X_i^k \quad\quad\quad
        (\phi^k_i)^{\dag} = \frac{1}{2}\left( 1 + Z_i^k \right) X_i^k \\
    \end{split}
    \end{equation}
    It is easy to prove the commutation relations to see that those operators are the creation and annihilation operators corresponding to hardcore bosons:
    \begin{equation}
    \begin{split}
        &\{ \phi_i^k, \phi_i^k \} = \{ \phi_i^{k~\dagger}, \phi_i^{k~\dagger} \} = 0 
        ~~~~~~~~
        \{ \phi_i^k , \phi_i^{k~\dagger} \} = 1 \\
        & [ \phi_i^k, \phi_j^l] = [\phi_i^{k~\dagger}, \phi_j^{l~\dagger}] = [\phi_i^k, \phi_j^{l~\dagger}] = 0 ~~~~~~~~~ \forall i\ne j, \forall k\ne l \\
    \end{split}
    \end{equation}
    Then, we can invert the relation and isolate operators that appear in the Hamiltonian:
    \begin{equation}
        Z_i^k = 2\phi_i^{k~\dagger}\phi_i^k - 1 \quad X_i^k = \phi_i^{k~\dagger} + \phi_i^k
    \end{equation}
    Now those operators can be substituted in the Hamiltonian $H$ and in this way we find the Hamiltonian written in terms of only hardcore bosonic degrees of freedom.
\end{proof}

As before, we now show the result for the $1$ dimensional case, but the full calculations both for the $1$ dimensional and $2$ dimensional cases are shown in Appendix \ref{bosons_2D}. An advantage of the construction is that this derivation works for every boundary condition (by adapting the corresponding Gauss' laws) and every spatial dimension $d$ of the lattice, which means that we will always be able to write the Hamiltonian in terms of hardcore bosonic degrees of freedom, keeping a local Hamiltonian (without all-to-all interactions). In $1$ dimension we define the hardcore bosonic operators:
\begin{equation}
\begin{split}
    \phi_l &= \frac{1}{2} \left( 1 - Z_l \right) X_l \quad\quad\quad
    \phi^{\dagger}_l = \frac{1}{2} \left( 1 + Z_l \right) X_l \\
    \label{boson_encoding}
\end{split}
\end{equation}
Then, by defining the number operator of bosons as $N_l = \phi_l^{\dagger} \phi_l$, we can write the Hamiltonian in terms of these new degrees of freedom (full calculations in Appendix \ref{bosons_2D}):
\begin{equation}
\begin{split}
    H &= 2m \sum_l (-1)^{l+1} N_{l-1} N_{l} + \epsilon \sum_l (N_{l-1}-N_{l+1})^2 (\phi_{l} + \phi_{l}^{\dagger}) + 2\lambda_E \sum_l (2N_l-1) \;,
\end{split}
\end{equation}
where we used the fact that $N_l =N_l^2$ for hard-core bosons.

As mentioned before, this form of the Hamiltonian underlines that sites have been integrated out since we used Gauss' law to write everything only as a function of the bosonic degrees of freedom. In Appendix \ref{bosons_2D} we show the full calculations for the $2$ dimensional case following the same procedure. In the same appendix, we also show another possible choice of logical operations which is non-local and does not lead to bosonic degrees of freedom but could be interesting for other reasons, since it could eliminate some non-locality given by the Jordan-Wigner representation. It is important to notice that there are two notions of locality here: the first one, of the physical interaction, is the one we claim to preserve. When integrating out matter sites, the new interaction are smeared out from the use of Gauss' law but are still spatially local. The second notion of locality is the locality of the encoded operators in the qubit space used for the embedding. In particular, the encoding of fermionic degrees of freedom into qubit could lead to non-local qubit operators as is the case with the Jordan-Wigner mapping used here. This non-locality is however independent of the particular gauge theory used, and in fact can be present even when there are no gauge bosons. Our procedure for integrating out the matter sites using Gauss' law preserves the qubit locality of the underlying qubit Hamiltonian in a way that is independent of the particular fermionic encoding. This means that, by using a different encoding from Jordan-Wigner, it could be possible to have also a local encoded Hamiltonian. In summary, the locality of the logical operations in our error correcting code is controlled only by the Gauss' law which is local, while the locality of the Hamiltonian of the system depends on the fermionic encoding employed and is preserved by our construction.

There are other ways to transform fermionic degrees of freedom into hardcore bosons, as done e.g. in Refs.\cite{Verstraete_2005,PhysRevB.98.075119}, and starting from a Hamiltonian already written in terms of hardcore bosons could simplify the $2$ dimensional case since it could avoid the non-local Pauli strings of Jordan-Wigner.

\section{Time evolution}
\label{section:5}

In previous sections, we built an error-correcting code exploiting the Gauss' Law. However, having an error-correcting code is not enough, if it cannot be used to perform operations. In this section, we will show how that code can be used to perform fault-tolerant operations on the system. In particular, we will provide a method to perform a universal gate set by state-injection, and methods to do a fault-tolerant time evolution of the Hamiltonian presented in the first section by using the error-correcting code described by Lemma \ref{lemma:QLDPC}. First, the Quantum Signal Processing (QSP) algorithm~\cite{PhysRevLett.118.010501,Low2019hamiltonian} will be used to prove that, given an ancilla register with a logarithmic number of (logical) qubits, all non-Clifford operations can be performed on the ancilla qubits, applying only Clifford operations on the error-correcting code. Then, we will discuss a method that uses product formulas, providing a fault-tolerant implementation and proving that a time evolution of $H$ with product formulas can be performed using at most $3$ logical ancilla qubits (Theorem~\ref{theorem:clifford_only}).

\subsection{Quantum Signal Processing}

Let us consider a time evolution via the QSP algorithm. We will use the linear combination of unitaries (LCU) algorithm~\cite{10.5555/2481569.2481570} to implement the Hamiltonian using an ancilla register. Then, from LCU to QSP one only needs to add another ancilla qubit, but for now, let us consider LCU only. For simplicity, let us look at the $1$-dimensional Hamiltonian, and recall $H$ in terms of the logical operations from Theorem \ref{theorem:H_logical_operations}:
\begin{equation}
    H = -\frac{m}{2}\sum_{l=1}^N (-1)^l Z_{l-1}Z_{l} + \frac{\epsilon}{2}\sum_{l=1}^N X_{l} - \frac{\epsilon}{2}\sum_{l=1}^N Z_{l-1} Z_{l+1}X_{l} + 2\lambda_E \sum_{l=1}^N Z_l
\end{equation}
To use the LCU algorithm, we need the normalized Hamiltonian $\tilde H = H/(N\eta)$ where $N$ is the number of links in the system and $\eta=\frac{m}{2} + \epsilon + 2\lambda_E$:
\begin{equation}
    \widetilde H = \frac{1}{N}\sum_{l=1}^N \left( \frac{m}{2\eta}(-1)^{l+1} Z_{l-1} Z_{l} + \frac{\epsilon}{2\eta}X_{l} + \frac{\epsilon}{2\eta}(-Z_{l-1} Z_{l+1}X_{l}) + \frac{2\lambda_E}{\eta}(-Z_l) \right)
\end{equation}
In general, considering the $d$ dimensional lattice with $N$ sites, we can always write the Hamiltonian as a sum over the $dN$ degrees of freedom (represented in our code by the links of the lattice) and a sum over the $K$ different operators with their coupling constants:
\begin{equation}
    \tilde H = \frac{H}{\eta dN} = \frac{1}{dN} \sum_{l=1}^{dN} \left( \sum_{k=0}^K \frac{\eta_k}{\eta} O_{k,l} \right)
\end{equation}
where $\eta = \sum_k \eta_k$. Note that we included the minus signs inside the operators so that all coupling constants are positive. From here, we can define the prepare and select operators:
\begin{equation}
\begin{split}
    \text{PREP}\ket{0}_a = \left( \frac{1}{\sqrt{2^n}} \sum_{l=0}^{2^n-1} \ket{l}_{a_1}\right) \otimes \left( \sum_{k=0}^K \sqrt{\frac{c_k}{C}} \ket{k}_{a_2} \right)
\end{split}
\end{equation}
\begin{equation}
    \text{SEL} = \sum_{k=0}^{K}  \ket{k}_{a_2}\bra{k}_{a_2} \left( \sum_{l=0}^{dN-1}\ket{l}_{a_1}\bra{l}_{a_1} \otimes O_{k,l} + \sum_{l=dN}^{2^n-1}\ket{l}_{a_1}\bra{l}_{a_1} \otimes \mathbb{I} \right)
    \label{eq:select}
\end{equation}
The $a$ subscript denotes the ancilla register, divided into the $a_1$ and $a_2$ registers, while $n=\lceil\log_2(dN)\rceil$ is the smallest integer such that $2^n \geq dN$. In general, we should use $dN$ instead of $2^n$ in the prepare operator, but having a power of $2$ makes everything much easier. So we use the smallest power of $2$ larger than $dN$, paying the price of implementing $\frac{dN}{2^n}\widetilde H$ instead of just $\widetilde H$. Since $\frac{1}{2}<\frac{N}{2^n}<1$, it means that we will have to at most fator of two in the normalization.

As for the size of the ancilla registers, $a_2=\lceil\log_2(K)\rceil$ is a small register (in the $1$ dimensional case it is only $2$ qubits, while in the $2$ dimensional case it will be $3$ qubits), while $a_1$ has the number of qubits equal to $n$ (which scales logarithmically in the number of sites of the system). The prepare operator is easy to do since it will be made of $n$ Hadamards applied on $a_1$, while on $a_2$ we need arbitrary state preparation, but we do not expect the size of this register to grow with the system size.

The select operator instead is more complicated, and we can estimate its complexity in terms of the number of Toffoli gates it requires. Given a circuit made of NOT, CNOT, and Toffoli (a so-called reversible NCT circuit), one can show that the lower bound to the number of Toffoli gates required to implement the select circuit of Equation~\ref{eq:select} is $2^n - n - 1$ (\cite{lower_bound_2016}, Lemma 4). The upper bound to the number of Toffoli gates instead is given by $1.5 \times 2^n - 4$ (\cite{Childs_2018}, Lemma G.7), using $n-1$ ancilla qubits.

We can now use the LCU oracles defined above to perform time-evolution using QSP by adding an additional qubit to the ancilla register which will be used to perform rotations and a controlled version of the SEL operation. All Toffoli gates and non-Clifford operations are done on the ancilla registers, which will have their own error-correcting code, while on our system will be applied only the controlled-$O_{k,l}$ operators. By definition, the $O_{k,l}$ are Pauli operators and the control will be a single qubit of the logical ancilla register, which means that we need to apply to the lattice Clifford operations only.

To ensure that those Clifford operations are fault-tolerant, let us assume that the ancilla qubit from which we control all the controlled-$O_{k,l}$ gates is in the $7$-qubit Steane code \cite{SteaneCode, Gottesman1997}. This code has logical operations $\overline X = X_1X_2X_3$ and $\overline Z = Z_1Z_2Z_3$, which are the same logical operations of the phase-flip repetition code used on the lattice. Since they are both CSS codes with the same logical operations, it means that the CNOT is transversal between the two codes, as shown in the left panel of Figure \ref{fig:transversal_CNOT_crosscode}. So, if $O_{k,l}$ is a weight $1$ operator for the Gauss' Law code, it is transversal and fault-tolerant by definition on the repetition code. If it has higher weight, it means that we have to repeat several times the circuit on the left panel of Figure \ref{fig:transversal_CNOT_crosscode}, and if an $X$ error happens on the control qubit it will spread on the system. To prevent this, the most trivial solution is to perform an error-correcting cycle on the ancilla register between every controlled logical operation, as shown in the right panel of Figure \ref{fig:transversal_CNOT_crosscode}. An alternative, that could be more efficient, is to use flag qubits to catch specific errors that are problematic \cite{Chamberland_2018, Chao_2020}.

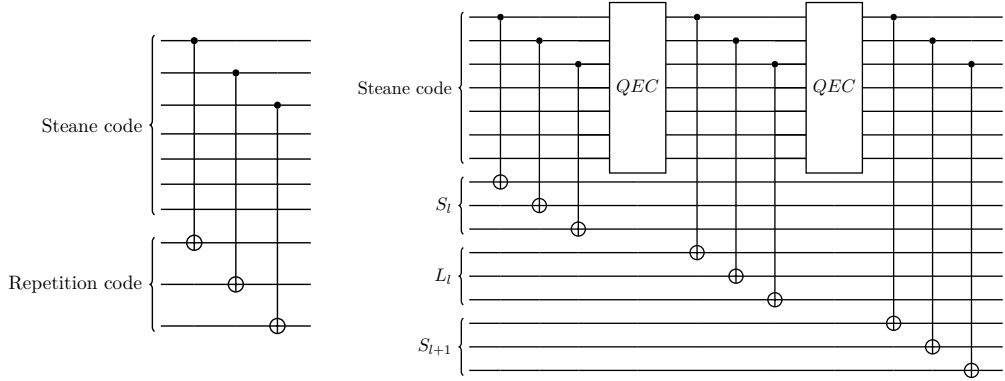
\begin{figure}[t]
\centering
\resizebox{0.3\textwidth}{!}{
\begin{quantikz}
    \lstick[7]{Steane code} & \ctrl{7} & \qw & \qw & \qw & \\
    & \qw & \ctrl{7} & \qw & \qw & \\
    & \qw & \qw & \ctrl{7} & \qw & \\
    & \qw & \qw & \qw & \qw & \\
    & \qw & \qw & \qw & \qw & \\
    & \qw & \qw & \qw & \qw & \\
    & \qw & \qw & \qw & \qw & \\
    \lstick[3]{Repetition code} & \targ{} & \qw & \qw & \qw & \\
    & \qw & \targ{} & \qw & \qw & \\
    & \qw & \qw & \targ{} & \qw & \\
\end{quantikz}}
\resizebox{0.6\textwidth}{!}{
\begin{quantikz}[row sep={0.5cm,between origins}]
    \lstick[7]{Steane code} & \ctrl{7} & \qw & \qw & \gate[7]{QEC} & \ctrl{10} & \qw & \qw & \gate[7]{QEC} & \ctrl{13} & \qw & \qw & \qw & \\
    & \qw & \ctrl{7} & \qw & \qw & \ghost{QEC}\qw & \ctrl{10} & \qw & \qw & \ghost{QEC}\qw & \ctrl{13} & \qw & \qw & \\
    & \qw & \qw & \ctrl{7} & \qw & \ghost{QEC}\qw & \qw & \ctrl{10} & \qw & \ghost{QEC}\qw & \qw & \ctrl{13} & \qw & \\
    & \qw & \qw & \qw & \qw & \ghost{QEC}\qw & \qw & \qw & \qw & \ghost{QEC}\qw & \qw & \qw & \qw & \\
    & \qw & \qw & \qw & \qw & \ghost{QEC}\qw & \qw & \qw & \qw & \ghost{QEC}\qw & \qw & \qw & \qw & \\
    & \qw & \qw & \qw & \qw & \ghost{QEC}\qw & \qw & \qw & \qw & \ghost{QEC}\qw & \qw & \qw & \qw & \\
    & \qw & \qw & \qw & \qw & \ghost{QEC}\qw & \qw & \qw & \qw & \ghost{QEC}\qw & \qw & \qw & \qw & \\
    \lstick[3]{$S_l$} & \targ{} & \qw & \qw & \qw & \qw & \qw & \qw & \qw & \qw & \qw & \qw & \qw & \\
    & \qw & \targ{} & \qw & \qw & \qw & \qw & \qw & \qw & \qw & \qw & \qw & \qw & \\
    & \qw & \qw & \targ{} & \qw & \qw & \qw & \qw & \qw & \qw & \qw & \qw & \qw & \\
    \lstick[3]{$L_l$} & \qw & \qw & \qw & \qw & \targ{} & \qw & \qw & \qw & \qw & \qw & \qw & \qw & \\
    & \qw & \qw & \qw & \qw & \qw & \targ{} & \qw & \qw & \qw & \qw & \qw & \qw & \\
    & \qw & \qw & \qw & \qw & \qw & \qw & \targ{} & \qw & \qw & \qw & \qw & \qw & \\
    \lstick[3]{$S_{l+1}$} & \qw & \qw & \qw & \qw & \qw & \qw & \qw & \qw & \targ{} & \qw & \qw & \qw & \\
    & \qw & \qw & \qw & \qw & \qw & \qw & \qw & \qw & \qw & \targ{} & \qw & \qw & \\
    & \qw & \qw & \qw & \qw & \qw & \qw & \qw & \qw & \qw & \qw & \targ{} & \qw & \\
\end{quantikz}}
\caption{On the left-hand side, the implementation of a logical CNOT between the Steane code and the phase-flip $3$-qubits repetition code, which is transversal. On the right-hand side, the fault-tolerant implementation of a logical $X_l = X_{S_l} X_{L_l} X_{L_{l+ 1}}$.}
\label{fig:transversal_CNOT_crosscode}
\end{figure}

\subsection{Trotterization}

In the last section, we showed that, using a logarithmic number of ancilla qubits, we are able to restrict the application on our error-correcting code of only Clifford operations, and perform a fault-tolerant time evolution of the Hamiltonian. Now, we will extend this argument by proving that not only this is possible, but also that the same can be achieved with a constant number of ancilla qubits. Of course, both methods have their own advantages, and there could be cases where Trotter-based simulations have a better asymptotic scaling than QSP (see e.g.~\cite{Shaw2020quantumalgorithms,Rajput2022hybridizedmethods}).

\begin{theorem}
    Let $H$ be the Hamiltonian of Equation \ref{Full_Hamiltonian}, and consider the lattice encoded in the quantum error-correcting code described in Lemma~\ref{lemma:QLDPC}. When performing the time evolution of $H$ using ancilla qubits, it is possible to perform all non-Clifford operations on the ancilla qubits and apply to the lattice Clifford operations only. This can be done with at most $3$ ancilla qubits encoded in the $7$-qubits code.
    \label{theorem:clifford_only}
\end{theorem}

\begin{proof}
    Let us consider the trotterization method to perform the time evolution of the Hamiltonian. By using this method, and remembering that the Hamiltonian contains only Pauli operators, we will need to apply to the system exponentials like $e^{itP}$ where $P$ is a Pauli operator. Note that $e^{itP} = \cos(t) + i\sin(t)P$ for every $P$ in the Pauli group, and we can implement this sum by using LCU. To do so, consider an ancilla in the $7$-qubit code, already prepared in the superposition $\cos(t)\ket{0} - i\sin(t)\ket{1}$. We can use this ancilla to implement $e^{itP}$ or $e^{-itP}$ with probability $1/2$ by using the circuit in the right panel of Fig.~\ref{fig:trotter_exponential} (where the success, and the application of the right exponential, will happen when the measurement outcome will be $\ket{0}$). The ancilla qubit is considered to be encoded in the $7$-qubit code since we already showed before how to do perform the controlled-pauli operator in a falut-tolerant way using that error-correcting code. Then, in order to prepare the ancilla in the right superposition, we will need to perform state injection using another ancilla encoded in the $7$ qubit code as well. The circuit to perform state injection is shown in the left panel of figure \ref{fig:trotter_exponential}.
    
\begin{figure}[h]
    \centering
    \begin{quantikz}[row sep={0.9cm,between origins},column sep={0.3cm}]
       \lstick{$\ket{\psi}$} & \qwbundle{} & \qw & \qw & \targ{} & \meter{} \\
       \lstick{$\ket{0}$} & \qwbundle{} & \gate{H} & \gate{T} & \ctrl{-1} & \gate{SX}\vcw{-1} & \qw & \rstick{$T\ket{\psi}$} \\
    \end{quantikz}
    \begin{quantikz}[row sep={0.9cm,between origins},column sep={0.3cm}]
        \lstick{$\cos(t)\ket{0}-i\sin(t)\ket{1}$} & \ctrl{1} & \gate{H} & \meter{} \\
        \lstick{$\ket{\phi}$} & \gate{P} & \qw \rstick{$e^{itP}\ket{\phi}$} \\
    \end{quantikz}
    \caption{On the left-hand side the circuit performs state-injection on a $7$-qubit code. On the right-hand side, the circuit applies $e^{itP}$ on $\ket{\phi}$ using an ancilla in the right superposition. In both circuits every operation is transversal. The second circuit has a success probability of $1/2$.
    }
    \label{fig:trotter_exponential}
\end{figure}
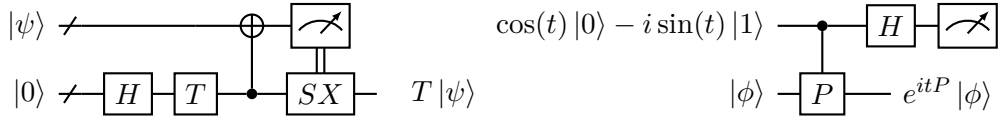
    
    With the circuit we described, we can implement every exponential of the type $e^{itP}$ with success probability of $1/2$, but we can use oblivious amplitude amplification \cite{Berry_2014} to make sure it always succeeds. To apply it, first we have to lower the success probability to $1/4$, which can be done by adding a qubit in the state $\frac{1}{\sqrt{2}}(\ket{0}+\ket{1})$ and by defining success the measured state to be $\ket{00}$ on the two ancilla qubits. So, let us define the $V$ operator as:
    \begin{equation}
    \begin{quantikz}
        \lstick{$\ket{0}$} & \gate[3]{V} & \qw \\
        \lstick{$\ket{0}$} & \ghost{V} & \qw \\
        \lstick{$\ket{\phi}$} & \ghost{V} & \qw \\
    \end{quantikz} = 
    \begin{quantikz}
    \lstick{$\ket{0}$} & \qw & \qw & \gate{H} & \qw \\
    \lstick{$\ket{0}$} & \gate{W(t)} & \ctrl{1} & \gate{H} & \qw \\
    \lstick{$\ket{\phi}$} & \qw & \gate{P} & \qw & \qw \\
    \end{quantikz}
    \label{V_operator}
    \end{equation}
    where the gate $W(t)$ performs the following rotation:
    \begin{equation}
        W(t)\ket{0} = \cos(t)\ket{0} - i\sin(t)\ket{1}
    \end{equation}
    So, as mentioned before, the application of the gate $V$ will success with probability $1/4$ when we measure the two ancilla qubits. Now, let us define the following projection and reflection operators:
    \begin{equation}
    \begin{split}
        \Pi_0 &= \ket{00}\bra{00}\otimes \mathbb{1} \\
        R &= 2\Pi_0 - \mathbb{1} \\
    \end{split}
    \end{equation}
    where the tensor product is ment to distinguish between the two ancilla qubits, and the system encoded in the state $\ket{\phi}$. Note that the reflection operator $R$ is nothing but a control $Z$ gate on the two ancilla qubits (which means it is also hermitian). Finally, we can define the operator for oblivious amplitude amplification $S = -VR V^{\dag}R$. With those definitions, if we apply the operator $SV=-VR V^{\dag} RV$, we will apply on the system (the $\ket{\phi}$ register) the right exponential $e^{itP}$ with probability $1$ (and using a single iteration of the operator $S$).
    
    In the end, we can implement every exponential coming from the trotterization method (fault-tolerantly) using at most $3$ logical ancilla qubits encoded in the $7$ qubit Steane code. Two are the ones of Equation \eqref{V_operator}, while the last one is the one needed to perform state injection and apply the gate $W(t)$.
\end{proof}

We just showed how one can apply $e^{itP}$ on the system fault-tolerantly for every $P$ in the Pauli group. However, even if we used it only to perform the time evolution of the Hamiltonian, this procedure can be used to achieve a universal gate set on our error-correcting code.

\begin{theorem}
    A universal gate set can be performed on the system encoded in the error-correcting code of Lemma \ref{lemma:QLDPC} using at most $3$ logical ancilla qubits initialized in the $7$-qubit code.
\end{theorem}

\begin{proof}
    Pauli operations are transversal as well as the CNOT, since the code is CSS. Then, since we are able to apply on the system $e^{i\theta P}$ for every angle $\theta$ and Pauli operator $P$, we can use this to apply operations that are not already transversal. It is trivial to perform arbitrary $Z$ rotations via their definition as $R_Z(\theta) = e^{i\theta Z}$, while we can perform the Hadamard gate by using the definition $H = XR_y(\frac{\pi}{2})$ and $R_y(\theta) = e^{i\theta Y}$.
\end{proof}

\section{Conclusions}
We extended the original idea of Ref.~\cite{Rajput2023} which builds quantum error-correcting codes by exploiting the redundancy given by the Gauge invariance in a Lattice Gauge Theory in several ways. Our first new contribution is the extension of the original proposal for a $\mathbb{Z}_2$ gauge theory in one and two dimensions to arbitrary spatial dimensions while being agnostic on the choice of the phase-flip error-correcting code employed to make the whole protocol fault-tolerant. This leads to either a logarithmic or linear overhead in the number of qubits to perform the error correction, depending whether we want local checks or not. We have then shown that by encoding the Hamiltonian with this error-correcting code, $H$ can be written in terms of the logical operations of the error-correcting code itself, which is equivalent to integrating out fermions. Using this procedure, it becomes easy to transform the original fermionic Hamiltonian into a hardcore-bosonic Hamiltonian. 

The technique of writing the Hamiltonian of a lattice gauge theory in terms of the logical operations allowed in the corresponding Gauss' Law code is a powerful and simple strategy to find dual theories to the original model. Then, we have provided an explicit fault-tolerant algorithm to perform a universal set of operations on our Gauss' Law code and use it to show how to perform Hamiltonian simulation. We described in detail simulations using both the Quantum Signal Processing and Trotterization methods, showing that only a logarithmic and constant number of ancilla qubits is required for the two methods respectively. We have achieved that by showing that, on the error-correcting code exploiting the Gauss' law, every operation needed for Hamiltonian simulation is Clifford and can be done transversally. If the computation will require non-Clifford operations, they will be always done on an ancilla register, that is assumed to have it's own error-correcting code.

Several extensions of the present work are possible including a complete characterization of the Gauss' Law code for larger discrete Abelian groups such as $\mathbb{Z}_N$ as well as the design of error-correcting codes tailored for non-Abelian discrete gauge groups. 
Non-Abelian theories cannot generate a stabilizer group, which is assumed to be an Abelian group, and so one way to solve the problem could be non-stabilizer error-correcting codes. As for increasing the gauge group dimension from $N=2$ to larger values, one can either add qubits to every link of the lattice, or work with qudits-links. The qudit case appears easier form the theoretical point of view, since the Gauss' Law remains in the generalized Pauli group. However, it is not clear what is the real noise model for multi-level systems, and so many questions have still to be addressed even in this case.

\acknowledgments{
We would like to thank P. Hauke, J. Mildenberger and A. Rajput for many useful discussions about the topics discussed in this work. A.R. is funded by the European Union. Views and opinions expressed are however those of the author(s) only and do not necessarily reflect those of the European Union or the European Commission. Neither the European Union nor the granting authority can be held responsible for them. This project has received funding from the European Union’s Horizon Europe research and innovation programme under grant agreement No.\ 101080086 NeQST.
\begin{center}
    \includegraphics[width=0.1\textwidth]{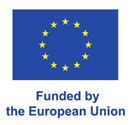}
\end{center}
}

\appendix

\section{Proof of Lemma \ref{lemma:asymptotic_optimality}}
\label{Other_phase_flip}

\begin{proof}[Proof of Lemma \ref{lemma:asymptotic_optimality}]
Consider the Quantum Hamming codes \cite{Steane_1996} with parameters
\begin{equation}
    [[2^r-1, 2^r-1-2r,3]]\;,
\end{equation}
with $r\geq3$. These are CSS quantum error-correcting code made by a concatenation of $2$ classical Hamming codes. If we consider $2^r-1-2r \geq N+dN$, we can use the logical qubits of this code as physical qubits for the classical error-correcting code of Theorem \ref{theorem_gauss_law} (that has parameters $[N+dN,dN,3]$). Since this second code is classical, it corrects only one type of errors by definition and it cannot increase the distance of the Hamming code. Now, in order to choose $r$ we need $2^r-1-2r \geq N+dN$.
We can take for instance
\begin{equation}
\label{eq:rlog}
r=\log_2(N+dN+1)+\log_2\left(1+\alpha\frac{\log_2(N+dN+1)}{N+dN+1}\right)\;,
\end{equation}
so that the condition we need is
\begin{equation}
(\alpha-2)\log_2(N+dN+1)-2\log_2\left(1+\alpha\frac{\log_2(N+dN+1)}{N+dN+1}\right)\geq 0\;.
\end{equation}
We can make sure this is satisfied for any $N\geq1$ by using $\alpha=6$.
The total number of physical qubits is thus
\begin{equation}
\label{eq:num_qub}
2^r-1 = N+dN+1+\alpha\log_2(N+dN+1)\;.
\end{equation}

The Quantum Hamming code can be written as:
\begin{equation}
    [[2^r-1, N+dN, 3]]
\end{equation}
and with the concatenation of the code from Theorem \ref{theorem_gauss_law} we get a code with parameters:
\begin{equation}
    [[2^r-1, dN, 3]]
\end{equation}
Finally, using the choice of $r$ from Eq.~\eqref{eq:rlog} and $\alpha=6$ together with Eq.~\eqref{eq:num_qub} we can write the final parameters of the code as
\begin{equation}
    [[N+dN + 6\log_2(N+dN+1), N+dN, 3]]
\end{equation}
Furthermore, since the concatenation involves a CSS code and a classical code, $X$ and $Z$ stabilizers will never be mixed, so the final code will be again CSS.

\end{proof}

It is possible to build more efficient codes with respect to the one described in the proof of Lemma \ref{lemma:asymptotic_optimality}. This is due to the fact that we do not need in general a full quantum error-correcting code to concatenate the code of Theorem \ref{theorem_gauss_law}. We could instead build a quantum error-correcting code by considering a classical Hamming code for $Z$ errors, and the code of Theorem \ref{theorem_gauss_law} for $X$ errors. This will not modify the scaling but it could improve the qubit overhead needed. However, it could happen in this case that not every error syndrome is unique, and we could need some additional qubits ($O(1)$) to make every error syndrome unique.

An example of such a code is shown in Table \ref{2sites_system_hamming}, where we are considering a $2$ sites $2$ links system with periodic boundary conditions. In this case, we used a classical Hamming code to encode qubits in the phase-flip code, while for the bit-flip we used Gauss' law operators, and only $2$ additional qubits to make every error syndrome of the bit-flip code unique.

\begin{table}[h]
    \centering
    \begin{tabular}{c|ccccccccc|cc}
        \hline
        & $q_1$ & $q_2$ & $q_3$ & $q_4$ & $q_5$ & $q_6$ & $q_7$ & $q_8$ & $q_9$ & $q_{10}$ & $q_{11}$ \\
        \hline
        $S_1$ & X & 1 & X & 1 & X & 1 & X & 1 & X & 1 & X \\
        $S_2$ & 1 & X & X & 1 & 1 & X & X & 1 & 1 & X & X \\
        $S_3$ & 1 & 1 & 1 & X & X & X & X & 1 & 1 & 1 & 1 \\
        $S_4$ & 1 & 1 & 1 & 1 & 1 & 1 & 1 & X & X & X & X \\
        \hline
        \hline
        $\bar Z_1 = A_0$ & Z & Z & Z & 1 & 1 & 1 & 1 & 1 & 1 \\
        $\bar Z_2 = L_0$ & 1 & Z & Z & Z & Z & 1 & 1 & 1 & 1 \\
        $\bar Z_3 = S_1$ & 1 & 1 & 1 & Z & Z & Z & Z & 1 & 1 \\
        $\bar Z_4 = L_1$ & 1 & 1 & 1 & 1 & 1 & Z & Z & Z & Z \\
        $\bar Z_5 = S_0$ & 1 & 1 & Z & Z & 1 & 1 & Z & 1 & 1 \\
        \hline
        $\bar X_1$ & X & 1 & 1 & 1 & 1 & 1 & 1 & 1 & 1 \\
        $\bar X_2$ & X & 1 & X & X & X & 1 & 1 & 1 & 1 \\
        $\bar X_3$ & 1 & 1 & 1 & X & X & 1 & X & X & 1 \\
        $\bar X_4$ & 1 & 1 & 1 & 1 & 1 & 1 & 1 & X & 1 \\
        $\bar X_5$ & 1 & X & X & 1 & 1 & 1 & 1 & 1 & 1 \\
        \hline
        \hline
        $G_1$ & 1 & Z & Z & 1 & 1 & 1 & 1 & Z & Z & 1 & 1 \\
        $G_2$ & 1 & Z & 1 & 1 & Z & Z & 1 & Z & Z & 1 & 1 \\
        $G_3$ & Z & 1 & 1 & Z & Z & 1 & 1 & 1 & 1 & 1 & 1 \\
        \hline
        $G_4$ & 1 & 1 & 1 & Z & 1 & 1 & Z & Z & 1 & 1 & Z \\
        $G_5$ & 1 & 1 & 1 & Z & 1 & 1 & Z & 1 & Z & Z & 1 \\
        \hline
    \end{tabular}
    \caption{Explicit representation of the Hamming phase-flip code, using a $2$ sites system. In this case, we need $2$ additional bit-flip stabilizers to make the code distance $3$. The $S_i$ operators are the phase-flip stabilizers, while $G_i$ labels the bit-flip stabilizers. $G_i$ with $i\le 3$ are built by multiplying logical operations as in the definition of the Gauss' law.}
    \label{2sites_system_hamming}
\end{table}

\section{Bosonic encoding}
\label{bosons_2D}
In this section we report the proof of Theorem~\ref{theorem:H_logical_operations}, as well as the full calculations to write the $1$ and $2$ dimensional Hamiltonian as a function of bosonic degrees of freedom integrating out the fermions. At the end of the section, we also show another possible choice for the logical operations of the error-correcting code, which is non-local and it's harder to deal with. However, we believe that those non-localities raising when writing $H$ in terms of logical operations could cancel non-localities coming from Jordan-Wigner in the $2$ dimensional case, or any way they could be advantageous in some cases.

\subsection{Proof of Theorem \ref{theorem:H_logical_operations}}
\begin{proof}[Proof of Theorem \ref{theorem:H_logical_operations}]
Remember the Hamiltonian of Eq.~\eqref{Full_Hamiltonian}:
\begin{equation}
    H = H_{M} + H_{hop} + H_{E} + H_{P}
\end{equation}
\begin{equation}
\begin{split}
    H_{M} &= m \sum_{\vec{l}} \sigma_{\vec{l}}\; \psi_{\vec{l}}^{\dagger}\psi_{\vec{l}}\\
    H_{hop} &= \epsilon \sum_{\vec{l}}\sum_{\hat k\in \mathcal{D}} \sigma_{\vec{l},\hat k} \left(\psi_{\vec{l}}^{\dagger} Q_{\vec{l},\hat k} \psi_{\vec{l}+\hat k} + \psi_{\vec{l}+\hat k}^{\dagger} Q^{\dagger}_{\vec{l},\hat k} \psi_{\vec{l}}\right)\\
    H_{E} &= \lambda_E \sum_{\vec{l}} \sum_{\hat k\in \mathcal{D}}\left( P_{\vec{l},\hat k} + P_{\vec{l},\hat k}^{\dagger} \right)\\
    H_{M} &= \lambda_P \sum_{p} (W_{p} + W^{\dagger}_{p})
\end{split}
\end{equation}

Since the stabilizer group is generated by the Gauss' Law operator, the Hamiltonian will commute with it. So, the fact that the Hamiltonian can always be written in terms of the logical operations follows trivially since the Gauss' Law and the logical operations are a complete basis of the centralizer of the stabilizer group (set of elements of the Pauli group that commute with the every element of the stabilizer group).

Let us now define the Jordan-Wigner transformation~\cite{Jordan:1928wi}. It is a method to write the fermionic creation and annihilation operators in term of the Pauli matrices, preserving the commutation relations. First, we can define the creation and annihilation operators as the operators $\rvert1\rangle\langle0\lvert$ and $\rvert0\rangle\langle1\lvert$ respectively for even sites, while for odd sites we invert the definition. In terms of Pauli matrices we can write:
\begin{equation}
\begin{split}
    a_{\vec{l}} &= \frac{1}{2} \left(1 + Z_{\vec{l}} \right) X_{\vec{l}} \quad\quad\quad
    a_{\vec{l}}^{\dagger} = \frac{1}{2} \left(1 - Z_{\vec{l}} \right) X_{\vec{l}} \;.
\end{split}
\end{equation}
Looking at the same site, those operators have the right commutation relations:
\begin{equation}
    \{ a_{\vec{l}}, a_{\vec{l}} \} = \{ a_{\vec{l}}^{\dagger}, a_{\vec{l}}^{\dagger} \} = 0 \quad\quad\quad \{ a_{\vec{l}}, a_{\vec{l}}^{\dagger} \} = 1\;.
\end{equation}
However, operators on different sites commute, while fermionic operator should anticommute. To adjust those relations, let us assume for simplicity a $1$-dimensional lattice, we can define the proper fermionic creation and anihilation operators as follows:
\begin{equation}
\begin{split}
    \psi_{l} &= \left( \prod_{i = 0}^{l-1} -Z_{i} \right) a_l \quad\quad\quad
    \psi_{l}^\dag = \left( \prod_{i = 0}^{l-1} -Z_{i} \right) a_l^\dag \\
\end{split}
\end{equation}
In higher spatial dimensions, the definition is straightforwardly generalized by defining an ordering so that we can go through every site with a single index $n(\vec{l})$, one for each of the $N$ sites. In this case the definition of the fermionic operators becomes:
\begin{equation}
\begin{split}
    \psi_{\vec{l}} &= \left( \prod_{k = 0}^{n\left(\vec{l}\right)-1} -Z_{\vec{l}_k} \right) a_{\vec{l}} \quad\quad\quad
    \psi_{\vec{l}}^\dag = \left( \prod_{k = 0}^{n(\vec{l})-1} -Z_{\vec{l}_k} \right) a_{\vec{l}}^\dag \;,
\end{split}
\end{equation}
where $k=n(\vec{l}_k)$. With this definition, it is easy to show that these operators satisfy the right commutation relations:
\begin{equation}
    \{ \psi_{\vec{l}}, \psi_{\vec{m}} \} = \{ \psi_{\vec{l}}^\dag, \psi_{\vec{m}}^\dag \} = 0 \quad\quad\quad \{ \psi_{\vec{l}}, \psi_{\vec{m}}^\dag \} = \delta_{\vec{l}, \vec{m}}\;.
\end{equation}
Using the Jordan-Wigner construction and the $\mathbb{Z}_2$ symmetry we can write $H$ as follows:
\begin{equation}
\begin{split}
    H &= \frac{m}{2} \sum_{\vec{l}} \sigma_{\vec{l}} (1 -(-1)^{|\vec{l}|} Z_{S_{\vec{l}}}) \\
    &+ \frac{\epsilon}{2} \sum_{\vec{l}}\sum_{\hat k\in \mathcal{D}} \sigma_{\vec{l},\hat{k}} O_{\vec{l},\hat{k}} \left( 1 + Z_{S_{\vec{l}}}Z_{S_{\vec{l}+\hat k}} \right) X_{S_{\vec{l}}}\,X_{L_{\vec{l},\hat k}}X_{S_{\vec{l}+\hat k}} \\
    &+ 2\lambda_E\sum_{\vec{l}}\sum_{\hat k\in \mathcal{D}} Z_{L_{\vec{l},\hat k}} +2\lambda_P\sum_{p} X_{p_1} X_{p_2}  X_{p_3} X_{p_4} \\
\end{split}
\end{equation}
where $O_{\vec{l},\hat{k}}$ is the Pauli string given by Jordan-Wigner.
The last term sums over all plaquettes, and the numbering convention is according to Fig.~\ref{fig:plaquetteConvention} as explained in the first chapter of the main text. Recall that, for a $d$-dimensional lattice, the logical operations of the error-correcting code of Theorem~\ref{theorem_gauss_law} are
\begin{equation}
    \overline{Z}_{\vec{l},\hat{k}} = Z_{L_{\vec{l},\hat{k}}} \quad\quad\quad \overline{X}_{\vec{l},\hat{k}} = X_{S_{\vec{l}}}X_{L_{\vec{l},\hat{k}}}X_{S_{\vec{l}+\hat k}}
\end{equation}
while the Gauss' law operator, which for code-words is nothing but the logical identity, is
\begin{equation}
    G_{\vec{l}} = (-1)^{|\vec{l}|} Z_{S_{\vec{l}}}\prod_{\hat k \in \mathcal{D}} Z_{L_{\vec{l},\hat{k}}} Z_{L_{\vec{l}-\hat{k},\hat{k}}}
\end{equation}
From here we can write the Hamiltonian $H$ in terms of these logical operations. First, note that $X$ operators in the hopping term appear already in the form of the logical operation $\overline{X}_{\vec{l},\hat{k}}$, so we can simply substitute it. The $Z$ operators acting on links they translate directly into logical $\overline{Z}_{\vec{l},\hat{k}}$, while for the $Z$ on sites we use Gauss' law to write
\begin{equation}
Z_{S_{\vec{l}}}=(-1)^{|\vec{l}|}G_{\vec{l}}Z_{L_{\vec{l},\hat{k}}} Z_{L_{\vec{l}-\hat{k},\hat{k}}}=(-1)^{|\vec{l}|}Z_{L_{\vec{l},\hat{k}}} Z_{L_{\vec{l}-\hat{k},\hat{k}}}\;.
\end{equation}
The plaquette operator can be obtained by multiplication of $4$ logical operations, so that $X$ operations on sites cancel and we get the 4 $X$ on links. In this way we can write as logical operations every operator in the Hamiltonian. As for the degrees of freedom, we started with a degree of freedom for every link and one for every site (which are $N+dN$), while now the Hamiltonian is written in function of the logical qubits (which we know from Theorem~\ref{theorem_gauss_law} they are $dN$).
\end{proof}

\subsection{Full calculation for the 1-dimensional case}
Let us start from the Hamiltonian written in terms of logical operations
\begin{equation}
    H = \frac{m}{2}\sum_{l} (-1)^l(1 - Z_{l-1} Z_{l}) +\frac{\epsilon}{2} \sum_l \left( 1 -  Z_{l-1} Z_{l+1} \right) X_{l} +2\lambda_E \sum_l Z_l
\end{equation}
and the hardcore bosonic operators
\begin{equation}\begin{split}
    \phi_l &= \frac{1}{2} \left( 1 - Z_l \right) X_l \quad\quad\quad
    \phi^{\dagger}_l = \frac{1}{2} \left( 1 + Z_l \right) X_l \\
\end{split}
\end{equation}
We define the number operator $N_l = \phi_l^{\dagger} \phi_l = \frac{1}{2}(1+Z_l)$ and note that it is idempotent ($N_l = N_l^2$). So the various terms in the Hamiltonian can be written as:
\begin{equation}
    Z_l = 2N_l - 1
\end{equation}
\begin{equation}
\begin{split}
\sum_{l} (-1)^l(1 - Z_{l-1} Z_{l}) &= \sum_{l} (-1)^l(1 - (2N_{l-1}-1) (2N_{l}-1)) \\
    &= -4\sum_l (-1)^l N_{l-1} N_{l} + 2\sum_l (-1)^l N_l + 2\sum_l (-1)^{l-1}N_l \\
    &= -4\sum_l (-1)^l N_{l-1} N_{l} \\
\end{split}
\end{equation}
\begin{equation}
\begin{split}
1 - Z_{l-1} Z_{l+1} &= 1 - (2N_{l-1}-1) (2N_{l+1}-1) \\
    &= 2(N_{l-1} + N_{l+1} - 2N_{l-1} N_{l+1} ) \\
    &= 2(N_{l-1} - N_{l+1})^2 \\
\end{split}
\end{equation}
We note that the disappearance of the terms linear in $N_l$ from the second term is a direct consequence of using Periodic Boundary Conditions, in more general settings boundary terms might still be present. In the end, the Hamiltonian will be:
\begin{equation}
\begin{split}
    H &= 2m \sum_l (-1)^{l+1} N_{l-1} N_{l} \\
    &+ \epsilon \sum_l (N_{l-1}-N_{l+1})^2 (\phi_{l} + \phi_{l}^{\dagger}) +2 \lambda_E \sum_l (2N_l-1) \\
\end{split}
\end{equation}

\subsection{Full calculations for the 2-dimensional case}
Let us consider a lattice with $N_x$ by $N_y$ sites. We start by writing the Hamiltonian as a  function of the Pauli matrices, explicitly taking into account the fact that now Jordan-Wigner has non-local Pauli strings:
\begin{equation}
\begin{split}
    H &= \frac{m}{2} \sum_{i,j} (-1)^{i+j}(1 -(-1)^{i+j} Z_{S_{i,j}}) \\
    &+ \frac{\epsilon}{2} \sum_{i,j} \left( 1 + Z_{S_{i,j}}Z_{S_{i,j+1}} \right) X_{S_{i,j}}X_{L_{i,j}^{x}}X_{S_{i,j+1}} \\
    &+ \frac{\epsilon}{2} \sum_{i,j} (-1)^{j} P_{i,j} \left( 1 + Z_{S_{i,j}}Z_{S_{i+1,j}} \right) X_{S_{i,j}}X_{L_{i,j}^y} X_{S_{i+1,j}} \\
    &+ 2\lambda_E\sum_{i,j,k} Z_{L_{i,j}^k} + 2\lambda_P\sum_{i,j} X_{L^x_{i,j}} X_{L^y_{i,j+1}}  X_{L^x_{i+1,j}} X_{L^y_{i,j}} \\
\end{split}
\end{equation}
where $P_{i,j} =\left( \prod_{n=j}^{N_x} Z_{S_{i,n}} \right) \left( \prod_{n=0}^{j} Z_{S_{i+1,n}} \right)$ is the Pauli string given by vertical links. In these expressions we used, for $\vec{l}=(i,j)$, the following definitions to ease the notation
\begin{equation}
\begin{split}    
S_{\vec{l}} = S_{i,j}\quad\quad L_{\vec{l},(0,1)} = L^{x}_{i,j}\quad\quad L_{\vec{l},(1,0)} = L^{y}_{i,j}\;.
\end{split}
\end{equation}
The logical operations are
\begin{equation}
\begin{split}
    Z^x_{i,j} = Z_{L^x_{i,j}}& \quad\quad Z^y_{i,j} = Z_{L^y_{i,j}} \quad\quad Z_{S_{i,j}} = (-)^{i+j} Z^x_{i,j} Z^y_{i,j} Z^x_{i,j-1} Z^y_{i-1,j}\\
    X_{i,j}^y &= X_{S_{i,j}} X_{L^y_{i,j}} X_{S_{i+1,j}} ~~~~~~~~ X_{i,j}^x = X_{S_{i,j}} X_{L^x_{i,j}} X_{S_{i,j+1}} \\
\end{split}
\end{equation}
and substituting we get:
\begin{equation}
\begin{split}
    H &= \frac{m}{2} \sum_{i,j} (-1)^{i+j}(1 - Z^x_{i,j} Z^y_{i,j} Z^x_{i,j-1} Z^y_{i-1,j}) \\
    &+ \frac{\epsilon}{2} \sum_{i,j} \left( 1 - Z^x_{i,j-1} Z^y_{i,j} Z^y_{i-1,j} Z^y_{i,j+1} Z^x_{i,j+1} Z^y_{i-1,j+1} \right) X_{i,j}^x \\
    &+ \frac{\epsilon}{2} \sum_{i,j} (-1)^{j} P_{i,j}\left( 1 - Z^x_{i,j-1} Z^y_{i-1,j} Z^x_{i,j} Z^x_{i+1,j} Z^y_{i+1,j} Z^x_{i+1,j-1} \right) X_{i,j}^y \\
    &- \lambda \sum_{i,j,k} Z_{i,j}^k + \frac{1}{\lambda}\sum_{i,j} X^x_{i,j} X^y_{i,j+1} X^x_{i+1,j} X^y_{i,j} \\
\end{split}
\end{equation}
This is the Hamiltonian in terms of the logical operations, that can be used to implement the time evolution. Then, we can define the hardcore boson operators in the same way as the $1$-dimensional case
\begin{equation}
\begin{split}
    \phi_{i,j}^k &= \frac{1}{2}(1 - Z_{i,j}^k) X_{i,j}^k \\
    \left(\phi_{i,j}^k\right)^{\dagger} &= \frac{1}{2}(1 + Z_{i,j}^k) X_{i,j}^k \\
\end{split}
\end{equation}
where $k=x,y$. Defining again $N^k_{i,j} = \phi_{i,j}^{k,\dagger} \phi_{i,j}^{k}$ as the number operator we have the following relations
\begin{equation}
    Z_{i,j}^k = 2N^k_{i,j}-1 \quad X_{i,j}^k = \phi_{i,j}^k + \phi_{i,j}^{k,\dagger}
    \label{eq:to_substitute}
\end{equation}
Let us look at the mass term:
\begin{equation}
\begin{split}
    &\sum_{i,j} (-1)^{i+j}(1 - Z^x_{i,j} Z^y_{i,j} Z^x_{i,j-1} Z^y_{i-1,j}) = \\
    &= \sum_{i,j} (-1)^{i+j}[1 - (2N^x_{i,j}-1) (2N^y_{i,j}-1) (2N^x_{i,j-1}-1) (2N^y_{i-1,j}-1)]
\end{split}
\end{equation}
All the other terms can be found by substituting into the Hamiltonian the relations of Equation \ref{eq:to_substitute}. In this way, we will have a Hamiltonian in terms of the bosonic operators and the dependence on sites is implicit, contained in the logical operations.

\subsection{Another choice of logical operations for the 1-dimensional case}

The logical operations can be chosen differently, and here we give an example for the one-dimensional case. However, in general, the previous choice was the easiest to deal with because is local. Indeed, we will see that the choice we will explain now is highly non-local, and also the new degrees of freedom will no longer be bosons.

Enumerating sites and links from $1$ to $N$, the logical operations are the following:
\begin{equation}
\begin{split}
    X_l &= X_{S_{1}} X_{S_l} \prod_{n=l}^{N}X_{L_n} ~~ \forall l>1, ~~~~~~~~ X_1 = \prod_{n=1}^N X_{L_n} \\
    Z_{l} &= (-1)^{l} Z_{S_{l}} ~~ \forall l>1, ~~~~~~~~ Z_{1} = Z_{L_1}
\end{split}
\end{equation}
which can be inverted using the stabilizers which will be then simplified, leading to:
\begin{equation}
\begin{split}
    Z_{L_l} &= \prod_{n=1}^{l} Z_{n} \\
    Z_{S_l} &= (-1)^{l} Z_{l} ~~~~~~~~ Z_{S_1} = \prod_{n=2}^N Z_{n} \\
    X_{S_{l}}X_{L_l}X_{S_{l+1}} &= X_{l} X_{l+1} ~~~~~~~~ X_{S_{N}}X_{L_N}X_{S_1} = X_N
\end{split}
\end{equation}

Now, let us define the following operators:
\begin{equation}
\begin{split}
    \phi_l &= \frac{1}{2}\left( 1 + \prod_{n=1}^l (-1)^{n-1}Z_n \right) X_l \\
    \phi_l^{\dagger} &= \frac{1}{2}\left( 1 - \prod_{n=1}^l (-1)^{n-1}Z_n \right) X_l \\
\end{split}
\end{equation}
which satisfy the following commutation relations
\begin{equation}
\begin{split}
    &\{\phi_l, \phi_l\} = 0 = \{\phi_l^{\dagger}, \phi_l^{\dagger}\} ~~~~~~~~ \{ \phi_l^{\dagger}, \phi_l \} = 1\\
    &\{ O_i, O_j \}_{i<j} = O_i (\phi_j^{\dagger} + \phi_j) \\
\end{split}
\end{equation}
Where $O_i$ is either $\phi_i$ or $\phi_i^{\dagger}$. With those definitions, it is easy to prove that:
\begin{equation}
\begin{split}
    Z_{S_l} &= 1-2\phi^\dagger_l \phi_l \\
    Z_{L_l} &= (-1)^l (1-2\phi_l^\dag \phi_l)(1-2\phi_{l-1}^\dag \phi_{l-1}) \\
    X_l &= \phi_l^\dag \phi_l \\
\end{split}
\end{equation}

From here, we can already see that everything is local, while substituting into the Hamiltonian we get:
\begin{equation}
\begin{split}
    H &= m\sum_{l=1}^N (-1)^l\frac{1}{2}(1-(-1)^lZ_{S_l}) + \epsilon \sum_{l=1}^N \frac{1}{2}(1+Z_{S_l}Z_{S_{l+1}})X_{S_l}X_{L_{l}}X_{S_{l+1}} + 2\lambda_E \sum_{l=1}^N Z_{L_l} \\
    &= m \sum_{l=1}^{N} \frac{1}{2}\left( (-1)^l - (1-2\phi_l^{\dagger}\phi_l)(1 - 2\phi_{l-1}^{\dagger}\phi_{l-1}) \right) \\
    &+ 2\epsilon \sum_{l=1}^N \left( 1- \phi_{l-1}^{\dagger}\phi_{l-1} - \phi_{l+1}^{\dagger}\phi_{l+1} + 2\phi_{l-1}^{\dagger}\phi_{l-1}\phi_{l+1}^{\dagger}\phi_{l+1}\right) (\phi_l^\dag + \phi_l)(\phi_{l+1}^\dag + \phi_{l+1}) \\
    &+ \lambda_E \sum_{l=1}^N (1 - 2\phi_l^{\dagger}\phi_l)
\end{split}
\end{equation}

One possible advantage of this choice of logical operations is that, since $\phi$ is non-local, in some cases it could be used to simplify some non-locality raising from other things (like Jordan-Wigner or other encodings).


\end{document}